\documentclass[a4paper, 11pt]{article}

\usepackage{geometry}
\usepackage{graphicx}
\usepackage{amsmath}
\usepackage{amssymb}
\usepackage{amsthm}
\usepackage{stmaryrd}
\usepackage{enumerate}
\usepackage{verbatim}
\usepackage{caption}
\usepackage{mathtools}
\usepackage{microtype}
\usepackage{tikz}
\usetikzlibrary{positioning,decorations.pathreplacing}

\newcommand*{\Nset}{\mathbb{N}}

\newcommand*{\di}{\Diamond}
\newcommand*{\bo}{\square}
\newcommand*{\F}{\mathcal{F}}
\newcommand*{\M}{\mathcal{M}}

\newcommand*{\A}{\mathcal{A}}
\renewcommand*{\AA}{\mathbb{A}}
\newcommand*{\B}{\mathcal{B}}

\newcommand*{\BB}{\mathbb{B}}
\newcommand*{\MM}{\mathrm{Mod}}
\newcommand*{\NN}{\mathbb{N}}

\newcommand*{\VV}{\mathbb{V}}
\newcommand*{\EE}{\mathbb{E}}
\newcommand*{\G}{\mathcal{G}}
\newcommand*{\T}{\mathcal{T}}
\newcommand*{\EF}[1]{\mathrm{FS}^\Phi_{#1}}
\newcommand*{\ML}{\mathrm{ML}}

\newcommand*{\FO}{\mathrm{FO}}
\newcommand*{\uTL}{\mathrm{unary\text{-}TL}}
\newcommand*{\PAL}{\mathrm{PAL}}
\newcommand*{\EL}{\mathrm{EL}}

\renewcommand*{\phi}{\varphi}
\renewcommand*{\theta}{\vartheta}

\newcommand*{\tower}{\mathrm{twr}}
\newcommand*{\dom}{\mathrm{dom}}
\newcommand*{\powerset}{\mathcal{P}}
\newcommand*{\kaikki}{\square}
\newcommand*{\seur}[2]{\bo(#1,#2)}
\newcommand*{\s}{\mathrm{sz}}
\newcommand*{\ms}{\mathrm{ms}}
\newcommand*{\cs}{\mathrm{cs}}
\newcommand*{\bisim}{\leftrightarroweq}
\newcommand*{\liitos}{\textstyle{\bigtriangleup}}
\newcommand*{\win}{(\mathrm{win})}
\newcommand*{\sep}{(\mathrm{sep})}
\newcommand*{\ordo}{\mathcal{O}}
\newcommand*{\CC}{\mathbb{C}}
\newcommand*{\DD}{\mathbb{D}}
\newcommand*{\Lit}{\mathit{Lit}}
\newcommand*{\LTL}{\mathrm{LTL}}

\newcommand*{\mnabla}{\mathop{\nabla}}
\newcommand*{\set}{V}
\newcommand*{\rel}{E}
\newcommand*{\back}{B}
\newcommand*{\lab}{\mathrm{lab}}
\newcommand*{\res}{\mathrm{res}}
\newcommand*{\lef}{\mathrm{left}}
\newcommand*{\rig}{\mathrm{right}}
\newcommand*{\tup}{c}
\newcommand*{\Log}{L}
\renewcommand{\AA}{\mathbb{A}}
\newcommand*{\AAA}{\mathfrak{A}}
\newcommand*{\BBB}{\mathfrak{B}}
\newcommand*{\MMM}{\mathfrak{M}}
\newcommand*{\FS}{\mu\textnormal{-}\mathrm{FS}^\Phi}

\newcommand*{\card}{\mathrm{card}}
\renewcommand{\phi}{\varphi}

\newcommand*{\lit}{l}
\renewcommand{\part}{\rightharpoonup}
\newcommand*{\size}{\mathrm{sz}}

\newcommand*{\new}{\mathsf{new}}
\newcommand*{\old}{\mathsf{old}}
\newcommand*{\CTL}{\mathrm{CTL}}
\newcommand*{\Var}{\mathit{Var}}
\newcommand*{\Prop}{\mathit{Prop}}

\newcommand*{\LL}{\mathbb{L}}
\newcommand*{\RR}{\mathbb{R}}
\newcommand*{\nice}{relevant}
\newcommand*{\poin}{\mathrm{PM}}
\newcommand*{\clock}{\mathrm{CM}}
\renewcommand{\L}{\mathrm{L}}

\newcommand{\Smash}[1]{\smash{\mathrlap[\scriptstyle]{#1}}}

\theoremstyle{plain}
\newtheorem{theorem}{Theorem}[section]
\newtheorem{lemma}[theorem]{Lemma}
\newtheorem{corollary}[theorem]{Corollary}

\theoremstyle{definition}
\newtheorem{definition}[theorem]{Definition}
\newtheorem{remark}[theorem]{Remark}

\begin{document}

\title{Formula size games for modal logic and $\mu$-calculus}
\author{Lauri Hella \\
Tampere University \\
FI-33014 \\
Tampere University \\
\texttt{lauri.hella@tuni.fi} \\
+358503053204
\and
Miikka Vilander \\
Tampere University \\
FI-33014 \\
Tampere University \\
\texttt{miikka.vilander@tuni.fi} \\
+358405677359 }

  \maketitle
  
\begin{abstract}
We propose a new version of formula size game for modal logic. The game characterizes
the equivalence of pointed Kripke-models up to formulas of given numbers of
modal operators and binary connectives. Our 
game is similar to the well-known Adler-Immerman game. However, due to a crucial difference
in the definition of positions of the game, its winning condition is simpler, and the 
second player does not have a trivial optimal strategy. Thus, unlike
the Adler-Immerman game, our game is a genuine two-person game.
We illustrate the use of the game by proving a non-elementary succinctness gap
between bisimulation invariant first-order logic $\FO$ and (basic) modal logic $\ML$. We also present a version of the game for the modal $\mu$-calculus $\L_\mu$ and show that $\FO$ is also non-elementarily more succinct than $\L_\mu$.
\end{abstract}



\section{Introduction}

Logical languages are often compared in terms of expressiveness and computational complexity. The authors of \cite{Gogic:1995:CLK:1625855.1625967} argue that another important semantic aspect of a logical 
language is the size of formulas needed for expressing properties of structures. 
If two logics $\Log$ and $\Log'$ are equivalent in terms of expressivity,
one of them may be able to express interesting properties much more succinctly than the other. According to the standard terminology, for a given function $f$ on natural numbers, $\Log$ is said to be $f$ times more succinct than $\Log'$ if there is a sequence $(\varphi_n)_{n\in\mathbb{N}}$, of $\Log$-formulas
such that for any sequence $(\psi_n)_{n\in\mathbb{N}}$ of equivalent 
$\Log'$-formulas, the size of $\psi_n$ is at least $f(m_n)$, where $m_n$ is the 
size of~$\varphi_n$.

The succinctness of various modal and temporal logics has been an active 
area of research for the last couple of decades; see e.g. 
\cite{ctl,lutzsattlerwolter,fo2,immermangames,markey,lutz} for earlier work
on the topic and \cite{frenchiliev,unionsuccinctness,syntaxtrees,vDitIliev,McC-DFrPiRe,Fe-DuIliev} for recent work. Typical results
in the area state an exponential succinctness gap between two equally expressive
logics. Often such a gap is reflected in the complexity
of the logics in question. For example, Etessami, Vardi and Wilke proved in 
\cite{fo2} that the two-variable fragment $\FO^2$ of first-order logic and $\uTL$ 
(a weak version of temporal logic) have the same expressive power over 
$\omega$-words, but $\FO^2$ is exponentially more 
succinct than $\uTL$. Furthermore, the complexity of satisfiability for $\FO^2$
is NEXPTIME-complete, while the complexity of $\uTL$ is in NP \cite{sistla}. 
However, being more succinct does not always imply higher complexity: 
for example, public announcement logic $\PAL$ is exponentially more succinct
than epistemic logic $\EL$, but the complexity of satisfiability is the same for both of them \cite{lutz}.

The most commonly used methods for proving succinctness results are
formula size games and extended syntax trees. Formula size games were first
introduced by Adler and Immerman in \cite{immermangames} for branching-time 
temporal logic $\CTL$. The method of extended syntax trees was originally
formulated by Grohe and Schweikardt in \cite{linearFO} for first-order logic. 
The notion of extended syntax tree was actually inspired by the Adler-Immerman 
game, and in a certain sense these two methods are equivalent: an extended syntax
tree can be interpreted as a winning strategy for one of the players of the 
corresponding formula size game.
Both of these methods have been adapted to a large number of modal languages,
including epistemic logic \cite{epistemic}, multimodal logics with union and intersection
operators on modalities \cite{multimodalsuccinctness} and modal logic with contingency
operators \cite{syntaxtrees}.

The basic idea of the Adler-Immerman game is that one of the players, S (spoiler), 
tries to show that two sets of pointed models $\AA$ and $\BB$ can be separated by
a formula of size $n$, while the other player, D (duplicator), aims to show that no 
formula of size at most $n$ suffices for this.
The moves that S makes in the game reflect directly the logical operators in a formula
that is supposed to separate the sets $\AA$ and $\BB$. Any pair $(\sigma,\delta)$ of 
strategies for the players S and D produces a finite game tree $T_{\sigma,\delta}$, and 
S wins this play if the size of $T_{\sigma,\delta}$ is at most $n$.
The strategy $\sigma$ is a winning strategy for S if using it, S wins every play of the game. 
If this is the case, then there is a formula of size at most $n$ that separates the sets, 
and this formula can actually be read from the strategy $\sigma$.

A peculiar feature of the Adler-Immerman game is that the second player, duplicator, 
can be completely eliminated from it. This is because D has an optimal strategy $\delta_{\max}$,
which is to always choose the maximal allowed answer; this strategy guarantees that 
the size of the tree $T_{\sigma,\delta}$ is as large as possible.  Thus, in this sense 
the Adler-Immerman game is not a genuine two-person game, but rather a one-person
game. Extended syntax trees, on the other hand, do away with the game aspect entirely.

In the present paper, we propose another type of formula size game for modal logic.
Our game is a natural adaptation of the game introduced by Hella and V\"a\"an\"anen 
\cite{formulasize} for propositional logic and first-order logic. 
The basic setting in our game is the same as in the Adler-Immerman game: there are
two players, S and D, and two sets of structures that S claims can
be separated by a formula of some given size. The crucial difference is that in our
game we define positions to be tuples $(k,\AA,\BB)$ instead of just pairs $(\AA,\BB)$
of sets of structures, where $k$ is a parameter referring to the number of 
modal operators and binary connectives in a formula. In each move S has to decrease
the parameter~$k$. The game ends when the players reach a
position $(k^*,\AA^*,\BB^*)$ such that either there is a literal separating 
$\AA^*$ and $\BB^*$, or S cannot make any moves because
$k^*=0$. In the former case, S wins the play; otherwise D wins.

Thus, in contrast to the Adler-Immerman game, to determine the winner in our game 
it suffices to consider a single ``leaf-node"
$(k^*,\AA^*,\BB^*)$ of the game tree. This also means that our game is a real two-person 
game: the final position $(k^*,\AA^*,\BB^*)$ of a play depends on the moves of D, and
there is no simple optimal strategy for D that could be used for eliminating the role of D
in the game. 

We believe that our game is more intuitive and thus, in some cases it may be easier to use than
the Adler-Immerman game. On the other hand, it should be remarked that the two games are 
essentially equivalent: The moves corresponding to connectives and modal operators
are the same in both games (when restricting to the sets $\AA$ and $\BB$ in a position
$(k,\AA,\BB)$). Hence, in principle, it is possible to translate a winning strategy in one
of the games to a corresponding winning strategy in the other. 

Additionally, we introduce a formula size game for modal $\mu$-calculus.
This game is obtained by adapting the formula size game of modal logic 
to the setting with fixed point operators $\mu$ and $\nu$. A new challenge
in defining such a game is that if S uses a fixed point $\eta X$ ($\eta\in\{\mu,\nu\}$)
as the logical operator in his move, and later uses the corresponding 
variable $X$, then in the next round, the game has to return to the
subformula that follows $\eta X$. This means that the play may become infinite,
and defining the correct winning condition for infinite plays is complicated.
We solve this problem by adding ordinal clocks to the pointed Kripke models in the 
sets $\AA$ and $\BB$. The idea is that the the ordinals corresponding to 
a fixed point variable $X$ decrease each time the game returns to an earlier 
formula from a position with label $X$. This, in conjunction with keeping the 
sets $\AA$ and $\BB$ always finite, guarantees that every play of
the game is finite. The idea of using ordinal clocks is also used in
\cite{boundedGTS} to define finite semantic games for 
$\Log_\mu$.

We illustrate the use of our games by proving two non-elementary succinctness gaps,
one between first-order logic $\FO$ and (basic) modal logic $\ML$ and the other between $\FO$ and modal $\mu$-calculus $\L_\mu$. More precisely, we define 
a bisimulation invariant property of pointed Kripke-models by a first-order formula 
of linear size, and show that this property cannot be defined by any $\ML$- or $\L_\mu$-formula
of size less than the exponential tower of height  $n-1$. 
Furthermore, we show that the same property of pointed Kripke-models is already
definable by a formula of size $\ordo(2^n)$ in a version $\ML^2$ of $2$-dimensional 
modal logic. Hence the same non-elementary succinctness result holds for $\ML^2$
over $\ML$.

A similar gap between $\FO$ and temporal logic follows from a construction in the PhD thesis
\cite{stockmeyer} of Stockmeyer. He proved that the satisfiability problem of $\FO$
over words is of non-elementary complexity. Etessami and Wilke \cite{untilhierarchy}
observed that from Stockmeyer's proof it is possible to extract $\FO$-formulas of size 
$\ordo(n)$ whose smallest models are words of length non-elementary in $n$.  
On the other hand, it is well known that any satisfiable formula of temporal logic
has a model of size $\ordo(2^n)$, where $n$ is the size of the formula. Another result related to ours can be found in \cite{otto}, where Otto shows that $\FO$ is exponentially more succinct than $\ML$ by relating the modal depth of the $\ML$-formula to the quantifier rank of the $\FO$-formula. In contrast to this, our proof relies entirely on the number of disjunctions and conjunctions in the $\ML$-formula.

For modal $\mu$-calculus, the literature regarding succinctness is scarcer. In \cite{GroSch2004} Grohe and Schweikardt show several succinctness gaps between monadic second-order logics, many with fixed points. They use automata-theoretic techniques and cite a non-elementary succinctness gap between MSO and $\L_\mu$ as well-known\footnote{We were unable to find a source in the literature for this result but we are reasonably convinced that a non-elementary gap already between $\FO$ and $\L_\mu$ is a new result.}.

The structure of the paper is as follows. In Section 2 we present the logics used in the paper, fix some notation and define our notion of formula size. In Section 3 we present the formula size game for $\ML$ and show some basic results for it. Section 4 is dedicated to the non-elementary succinctness gap between $\FO$ and $\ML$ and all necessary definitions and lemmas to prove it. In Section 5 we define the formula size game for $\L_\mu$ and show basic results. Finally, in Section 6 we show the non-elementary succinctness of $\FO$ over $\L_\mu$. Section 7 is the conclusion.

The work on modal logic was previously published in the conference paper \cite{DBLP:conf/aiml/HellaV16}. This version has some minor changes to the modal logic part and the sections on modal $\mu$-calculus are completely new.

\section{Preliminaries}

In this section we fix some notation, define the syntax and semantics of basic modal logic and modal $\mu$-calculus, and define our notions of formula size.
For a detailed account on the logics used in the paper, we refer to the textbook
\cite{blackburn} of Blackburn, de Rijke and Venema for basic modal logic and \cite{BRADFIELD2007721} for the modal $\mu$-calculus.

\subsection*{Basic modal logic and first-order logic}

Let $\Prop$ be an infinite set of proposition symbols and let $\Phi \subseteq \Prop$. Let $\M = (W, R, V)$, where $W$ is a set, $R \subseteq W \times W$ and $V : \Phi \to \powerset(W)$, and let $w \in W$. The structure $(\M, w)$ is called a \emph{pointed Kripke-model for $\Phi$}. 

Let $(\M, w)$ be a pointed Kripke-model. We use the notation
\[
\seur{\M}{w} := \{ (\M, v) \mid v \in W, wR^\M v\}.
\]
If $\AA$ is a set of pointed Kripke-models, we use the notation
\[
\kaikki\AA := \bigcup\limits_{(\M, w) \in \AA} \seur{\M}{w}.
\]
Furthermore, if $f$ is a function $f : \AA \to \kaikki\AA$ such that $f(\M, w) \in \seur{\M}{w}$ for every $(\M, w) \in \AA$, then we use the notation
\[
\di_f\AA := f(\AA).
\]

Intuitively $\bo(\M, w)$ is the set of all successor models of $(\M, w)$, $\bo\AA$ is the collection of all successor models of all models $(\M, w) \in \AA$ and $\di_f \AA$ consists of one successor for each model in $\AA$, where the successors are given by the function $f$.
We now define the syntax and semantics of basic modal logic for pointed models. 

Let $\Phi \subseteq \Prop$. The set of formulas of 
~$\ML(\Phi)$
is generated by the following grammar
\[
\varphi := \top \mid \bot \mid p \mid \neg p \mid (\varphi \land \varphi ) \mid ( \varphi 
\lor \varphi )
\mid \di \varphi \mid \bo \varphi,
\]
where $p \in \Phi$.

As is apparent from the definition of the syntax, we assume that all $\ML$-formulas are in negation normal form. This is useful for the formula size game that we introduce in the next section.

The satisfaction relation $(\M, w) \vDash \phi$ between pointed Kripke-models $(\M, w)$ and
$\ML(\Phi)$-formulas $\phi$ is defined as follows:
\begin{enumerate}[(1)]
\item $(\M, w) \vDash \top$ for all $(\M, w)$, and $(\M, w) \nvDash \bot$ for all $(\M, w)$,
\item $(\M, w) \vDash p \Leftrightarrow w \in V(p)$, and $(\M, w) \vDash \neg p \Leftrightarrow w \notin V(p)$,
\item $(\M, w) \vDash (\phi \land \psi) \Leftrightarrow (\M, w) \vDash \phi$ and $(\M, w) \vDash \psi$,
\item $(\M, w) \vDash (\phi \lor \psi) \Leftrightarrow (\M, w) \vDash \phi$ or $(\M, w) \vDash \psi$,
\item $(\M, w) \vDash \di\phi \Leftrightarrow$ there is $(\M, v) \in \seur{\M}{w}$ such that $(\M, v) \vDash \phi$,
\item $(\M, w) \vDash \bo\phi \Leftrightarrow$ for every $(\M, v) \in \seur{\M}{w}$ it holds that $(\M, v) \vDash \phi$.
\end{enumerate}
Furthermore, if $\AA$ is a class of pointed Kripke-models, then 
\[
\AA \vDash \phi \Leftrightarrow (\A, w) \vDash \phi\text{ for every }(\A, w) \in \AA.
\]
For the sake of convenience we also use the notation
\[
\AA \vDash \neg\phi \Leftrightarrow (\A, w) \nvDash \phi\text{ for every }(\A, w) \in \AA.
\]
Note that this is only a notational convention as $\neg\phi$ is not in negation normal form and as such is generally not a formula in our syntax.

In Section \ref{nelonen}, we consider the case $\Phi=\emptyset$. In this case the only available literals are the constants $\top$ and $\bot$, which are always true or false respectively.

The syntax and semantics for first-order logic are defined in the standard way. 
Each $\ML$-formula $\phi$ defines a class $\MM(\phi)$ of pointed Kripke-models:
\[
	\MM(\phi):=\{(\M,w)\mid (\M, w) \vDash \phi\}.
\] 
In the same way, any $\FO$-formula $\psi(x)$ in the vocabulary consisting of the accessibility
relation symbol $R$ and unary relation symbols $U_p$ for $p\in\Phi$ defines
a class $\MM(\psi)$ of pointed Kripke-models:
\[
	\MM(\psi):=\{(\M,w)\mid \M \vDash \psi[w/x]\}.
\] 
The formulas $\phi\in\ML$ and $\psi(x)\in\FO$ are \emph{equivalent} if $\MM(\phi)=\MM(\psi)$.

The well-known link between $\ML$ and $\FO$ is the following theorem.

\begin{theorem}[van Benthem Characterization Theorem]
A first-order formula $\psi(x)$ is equivalent to some formula in $\ML$ if and only if $\MM(\psi)$ is bisimulation invariant.
\end{theorem}

If a property of pointed Kripke-models is $n$-bisimulation invariant for some $n\in\NN$, then it is also bisimulation invariant.  Thus, $\FO$-definability and $n$-bisimulation invariance imply 
$\ML$-definability for any property of pointed Kripke-models. We will use this version of 
van Benthem's characterization in Section \ref{sec-property} for showing that a certain property
is $\ML$-definable.
For the sake of easier reading, we give here the definition of $n$-bisimulation.

\begin{definition}
Let $(\M, w)$ and $(\M', w')$ be pointed $\Phi$-models. We say that $(\M, w)$ and $(\M', w')$ are \emph{$n$-bisimilar}, $(\M, w) \bisim_n (\M', w')$, if there
are binary relations $Z_n \subseteq \dots \subseteq Z_0$ such that for every $0 \leq i \leq n-1$ we have
\begin{enumerate}[(1)]
\item $(\M,w)Z_n(\M', w')$,
\item if $(\M, v)Z_0(\M',v')$, then $(\M, v) \vDash p \Leftrightarrow (\M', v') \vDash p$ for each $p \in \Phi$,
\item if $(\M,v) Z_{i+1} (\M',v')$ and $(\M, u) \in \seur{\M}{v}$ then there is $(\M', u') \in \seur{\M'}{v'}$ such that $(\M, u) Z_i (\M', u')$,
\item if $(\M,v) Z_{i+1} (\M',v')$ and $(\M', u') \in \seur{\M'}{v'}$ then there is $(\M, u) \in \seur{\M}{v}$ such that $(\M, u) Z_i (\M', u')$.
\end{enumerate}
\end{definition}

It is well known that if $\Phi$ is finite, two pointed
$\Phi$-models are $n$-bisimilar if and only if they are equivalent with respect to $\ML(\Phi)$-formulas of modal depth at most $n$.

\subsection*{Modal $\mu$-calculus}

Let $\Phi \subseteq \Prop$ and let $\Var$ be an infinite set of variables. The syntax of the modal $\mu$-calculus $\L_\mu(\Phi)$ is given by the grammar:
\[
\phi ::= \top \mid \bot \mid p \mid \neg p \mid (\phi \lor \phi) \mid (\phi \land \phi) \mid \di \phi \mid \bo \phi \mid X \mid \mu X.\phi \mid \nu X.\phi,
\]
where $p \in \Phi$ and $X \in \Var$. Note that all formulas are again in negation normal form. We additionally assume for simplicity that variables of different fixed points are distinct.


Truth of formulas of $\L_\mu(\Phi)$ is, like $\ML$, evaluated on pointed Kripke models $(\M, w)$, where $\M = (W, R, V)$. Let $\phi \in \L_\mu(\Phi)$ and let $\rho : \Var \to \powerset(W)$ be a valuation of variables. We define truth relation $(\M, w) \vDash_\rho \phi$ between pointed models and $\L_\mu(\Phi)$-formulas. Let $\|\phi\|_\rho := \{w \in W \mid (\M, w) \vDash_\rho \phi\}$ and let $\Gamma_{\phi, \rho} : \powerset(W) \to \powerset(W)$ be an operator which maps $W'$ to $\|\phi\|_{\rho[W'/X]}$. The notation $\mathrm{LFP}$ stands for least fixed point of an operator and $\mathrm{GFP}$ for greatest fixed point. Since variables only occur positively in fixed point formulas, $\Gamma_{\phi, \rho}$ is a monotone operator. By the Tarski-Knaster Theorem, the least and greatest fixed points of such a monotone operator always exist. The recursive definition of $\vDash_\rho$ is as follows:
\begin{itemize}
	\item $(\M, w) \vDash_\rho p \Leftrightarrow w \in V(p)$,
	\item $(\M, w) \vDash_\rho X \Leftrightarrow w \in \rho(X)$,
	\item $(\M, w) \vDash_\rho (\phi \lor \psi) \Leftrightarrow (\M, w) \vDash_\rho \phi \text{ or } (\M, w) \vDash_\rho \psi$,
	\item $(\M, w) \vDash_\rho (\phi \land \psi) \Leftrightarrow (\M, w) \vDash_\rho \phi \text{ and } (\M, w) \vDash_\rho \psi$,
	\item $(\M, w) \vDash_\rho \di\phi \Leftrightarrow$ there is $(\M, v) \in \seur{\M}{w}$ such that $(\M, v) \vDash_\rho \phi$,
	\item $(\M, w) \vDash_\rho \bo\phi \Leftrightarrow$ for every $(\M, v) \in \seur{\M}{w}$ it holds that $(\M, v) \vDash_\rho \phi$,
	\item $(\M, w) \vDash_\rho \mu X.\phi \Leftrightarrow w \in \mathrm{LFP}(\Gamma_{\phi, \rho})$,
	\item $(\M, w) \vDash_\rho \nu X.\phi \Leftrightarrow w \in \mathrm{GFP}(\Gamma_{\phi, \rho})$.
\end{itemize}

\subsection*{Formula size}

We define notions of formula size for $\ML$, $\L_\mu$ and $\FO$. Note that many different notions are called formula size in the literature and our notion is close to the length of the formula as a string rather than, say, the DAG-size\footnote{The DAG-size of a formula $\phi$
is the number of edges of the syntactic structure of $\phi$ in the form of a directed acyclic graph. Thus since the fan-out in the DAG is at most two, the DAG-size is at most two times the number of subformulas of $\phi$.}of it.

\begin{definition}\label{mlsize}
	The \emph{size} of a formula $\phi \in \ML$, denoted $\size(\phi)$, is defined recursively as follows:
	\begin{enumerate}[(1)]
		\item If $\phi$ is a literal, then $\size(\phi) = 1$.
		\item If $\phi = \psi \lor \vartheta$ or $\phi = \psi \land \vartheta$, then $\size(\phi) = \size(\psi) + \size(\vartheta) + 1$.
		\item If $\phi = \di\psi$ or $\phi = \bo\psi$, then $\size(\phi) = \size(\psi) + 1$.
	\end{enumerate}
\end{definition}

\begin{definition}
	The \emph{size} of a formula $\phi \in \L_\mu$, denoted $\size(\phi)$, is defined recursively as follows:
	\begin{enumerate}[(1)]
		\item $\size(l) = \size(X) = 1$, where $l$ is a literal and $X$ is a variable,
		\item $\size(\phi \lor \psi) = \size(\phi \land \psi) = \size(\phi) + \size(\psi) + 1$,
		\item $\size(\di \phi) = \size(\bo \phi) = \size(\mu X.\phi) = \size(\nu X.\phi) = \size(\phi) + 1$.
	\end{enumerate}
\end{definition}

The size of a formula is essentially its length as a string. Note however, that we do not count negations as we view them as parts of literals. Another aspect worth mentioning is the size of descriptions of proposition symbols. If we have an infinite set of propositions, the size of the encoding of each symbol in a fixed size vocabulary necessarily grows logarithmically. Here we consider all propositions to be of size one.

Similarly we define formula size for $\FO$ to be the number of binary connectives, quantifiers and literals in the formula. In general this could lead to an arbitrarily
large difference between formula size and actual string length. For example if $f$ is a unary function symbol, then atomic formulas of the form $f(x) = x$, $f(f(x)) = x$ and so on, all have size 1. In this paper however, we only consider formulas with one binary relation so this is not an issue.

\begin{definition}
The \emph{size} of a formula $\phi \in \FO$, denoted by $\s(\phi)$, is defined recursively as follows:
\begin{enumerate}[(1)]
\item If $\phi$ is a literal, then $\s(\phi) = 1$.
\item If $\phi = \neg\psi$, then $\s(\phi) = \s(\psi)$.
\item If $\phi = \psi \lor \vartheta$ or $\phi = \psi \land \vartheta$, then $\s(\phi) = \s(\psi) + \s(\vartheta) + 1$.
\item If $\phi = \exists x\psi$ or $\phi = \forall x \psi$, then $\s(\phi) = \s(\psi) + 1$.
\end{enumerate}
\end{definition}

To refer to some rather large formula sizes we need the exponential tower function.

\begin{definition}
We define the function $\tower : \Nset \to \Nset$ recursively as follows:
\begin{align*}
\tower(0) &= 1 \\
\tower(n+1) &= 2^{\tower(n)}.
\end{align*}
We will also use in the sequel the binary logarithm function, denoted by $\log$.
\end{definition}

\subsection*{Separating classes by formulas}

The definitions of the formula size games in sections \ref{mlgamesection} and \ref{mugamesection} are based on the notion
of separating classes of pointed Kripke-models by formulas. Recall that by the notation $\BB \vDash \neg \phi$ we mean that for every model $(\B, w) \in \BB$, we have $(\B, w) \nvDash \phi$. As formulas of $\ML$ are also in $\L_\mu$, we only define the following for $\L_\mu$ and $\FO$.

\begin{definition}
Let $\AA$ and $\BB$ be classes of pointed Kripke-models. 
\\
(a) We say that a formula $\phi \in \L_\mu$ \emph{separates $\AA$ from $\BB$} if 
$\AA \vDash \phi$ and $\BB \vDash \neg\phi$. 
\\
(b) Similarly, a formula $\psi(x) \in \FO$ separates $\AA$ from $\BB$ if  
for all $(\M,w)\in \AA$, $ \M\vDash \psi[w/x]$ and for all $(\M,w)\in \BB$, $ \M\vDash \neg\psi[w/x]$.
\end{definition}

In other words, a formula $\phi\in \L_\mu$ separates $\AA$
from $\BB$ if $\AA\subseteq\MM(\phi)$ and $\BB\subseteq \overline{\MM(\phi)}$, where
$\overline{\MM(\phi)}$ is the complement of $\MM(\phi)$.

\section{The formula size game for ML}\label{mlgamesection}

As in the Adler-Immerman game, the basic idea in our formula size game is that
there are two players, S (Samson) and D (Delilah), who play on a pair $(\AA,\BB)$ of two
sets of pointed Kripke-models. The aim of S is to show that $\AA$ and $\BB$
can be separated by a formula with size
at most $k$, while D tries to refute this. The moves of S reflect the 
connectives and modal operators of a formula that is supposed to separate the sets. 

The crucial difference between our game and the Adler-Immerman game is that we define positions in the game to be tuples $(k,\AA,\BB)$ instead of just pairs $(\AA,\BB)$. As in the A-I game, D chooses for connective moves, which branch she would like to see played next. However, our game never returns to the branch not chosen, so D has a genuine choice to make. The winning condition of our game is based on a natural property of single positions instead of the size of the entire game tree. 

We give now the precise definition of our game.

\begin{definition}\label{peli}
Let $\AA_0$ and $\BB_0$ be sets of pointed $\Phi$-Kripke-models and let $k_0\in\NN$. \emph{The formula size game between the sets $\AA_0$ and $\BB_0$}, denoted
$\EF{k_0}(\AA_0, \BB_0)$, has two players, S and D. The number $k_0$ is the \emph{resource parameter} of the game. The
starting position of the game is $(k_0,\AA_0, \BB_0)$. Let the position after $n$ moves be $(k, \AA, \BB)$. If $k = 0$, D wins the game. If $k > 0$, S has the following five moves to
choose from:
\begin{itemize}
\item \emph{$\lor$-move}: First, S chooses natural numbers $k_1$ and $k_2$ and sets $\AA_1$ and $\AA_2$ such that 
$k_1+k_2+1 = k$ and $\AA_1 \cup \AA_2 = \AA$. Then D decides whether the game continues from the position $(k_1, \AA_1, \BB)$ or the position 
$(k_2, \AA_2, \BB)$.

\item \emph{$\land$-move}: First, S chooses natural numbers $k_1$ and $k_2$ and sets $\BB_1$ and $\BB_2$ such that 
$k_1+k_2+1 = k$ and $\BB_1 \cup \BB_2 = \BB$. Then D decides whether the game continues from the position
$(k_1, \AA, \BB_1)$ or the position $(k_2, \AA, \BB_2)$.

\item \emph{$\di$-move}: S chooses a function $f: \AA \to \kaikki\AA$ such that $f(\A, w) \in \seur{\A}{w}$ for all $(\A, w) \in \AA$
and the game continues from the position $(k-1, \di_f\AA, \kaikki\BB)$. 

\item \emph{$\bo$-move}: S chooses a function $g: \BB \to \kaikki\BB$ such that $g(\B, w) \in \seur{\B}{w}$ for all $(\B, w) \in
\BB$ and the game continues from the position $(k-1, \kaikki\AA, \di_g\BB)$. 

\item \emph{$\Lit$-move}: S chooses a literal $l \in \Lit(\Phi)$. If $l$ separates the sets $\AA$ and $\BB$, S wins. Otherwise D wins.
\end{itemize}
\end{definition}

Since D wins if $k$ runs out, the parameter $k$ can be thought of as a resource of S that she spends on connectives and literals. In addition if there is a model $(\M, w) \in \AA$ (or $\BB$) for which $\bo(\M, w) = \emptyset$, then S cannot make a $\di$- (or $\bo$-)move. 

We prove that the formula size game indeed characterizes the separation of 
two sets of pointed Kripke-models by a formula of a given size.

\begin{theorem}\label{peruslause}
Let $\AA$ and $\BB$ be sets of pointed $\Phi$-models and let $k$ be natural number. Then the following conditions are equivalent:
\begin{enumerate}[11111]
\item[$\win_k$] S has a winning strategy in the game $\EF{k}(\AA, \BB)$.
\item[$\sep_k$] There is a formula $\phi \in \ML(\Phi)$ such that $\s(\phi) \le k$ and the formula $\phi$ separates $\AA$ from $\BB$.
\end{enumerate}
\end{theorem}
\begin{proof}
The proof proceeds by induction on the number $k$. First assume $k = 1$. If S makes any non-literal move, D wins since $k = 0$ in the following position. So the only possibility for a winning strategy is a literal move. There is a winning literal move if and only if there is a literal which separates $\AA_0$ from $\BB_0$. Thus $\win_1 \Leftrightarrow \sep_1$.

Suppose then that $k > 1$ and $\win_l \Leftrightarrow \sep_l$ for all $l < k$. Assume first that $\win_{k}$ holds.
Consider the following cases according to the first move in the winning strategy of S. For $\lor$- and $\land$-moves we use the index $i$ to always mean $i \in \{1,2\}$.
\begin{enumerate}[(a)]
\item Assume the first move of the winning strategy is a literal move and $\phi$ is the literal chosen by S. Then $\phi$ separates $\AA$ and $\BB$ and $\s(\phi) = 1$ so $\sep_{k}$ trivially holds.

\item Assume that the first move of the winning strategy of S is a $\lor$-move choosing numbers $k_1, k_2 \in \Nset$ such that $k_1+k_2+1=k$, and sets $\AA_1, \AA_2 \subseteq \AA$ such that $\AA_1 \cup \AA_2 = \AA$. Since this move is given by a winning strategy, S has a
winning strategy for both possible continuations of the game, $(k_1, \AA_1, \BB)$ and $(k_2, \AA_2, \BB)$. Since $k_i < k$, by induction hypothesis there is a formula $\psi_i$ such that $\s(\psi_i) \le k_i$ and $\psi_i$ separates 
$\AA_i$ from $\BB$. Thus $\AA_i \vDash \psi_i$ so $\AA \vDash \psi_1 \lor \psi_2$. On the other hand $\BB \vDash \neg\psi_1$ and $\BB \vDash
\neg\psi_2$ so $\BB \vDash \neg(\psi_1 \lor \psi_2)$. Therefore the formula $\psi_1 \lor \psi_2$ separates $\AA$ from $\BB$. In addition
$\s(\psi_1 \lor \psi_2) = \s(\psi_1) + \s(\psi_2) + 1 \le k_1+k_2+1 = k$ so
$\sep_{k}$ holds.

\item The case in which the first move of the winning strategy of S is a right splitting move
is proved in the same way as the previous one, with the roles of $\AA$ and $\BB$
switched, and disjunction replaced by conjunction.

\item Assume that the first move of the winning strategy of S is a $\di$-move choosing a function $f: \AA \to \kaikki\AA$ such that $f(\A, w) \in \seur{\A}{w}$ for all
 $(\A, w) \in \AA$. The game continues from the position $(k-1, \di_f\AA, \kaikki\BB)$ and S has a winning strategy from this position. By induction 
hypothesis there is a formula $\psi$ such that $\s(\psi) \le k-1$ and $\psi$ separates $\di_f\AA$ from $\kaikki\BB$. Now for every $(\A, w)
\in \AA$ we have $f(\A, w) \in \seur{\A}{w}$ and $f(\A, w) \vDash \psi$. Therefore $\AA \vDash \di\psi$. On the other hand $\kaikki\BB \vDash \neg\psi$ so for every $(\B, w) \in \BB$ and every 
$(\B, v) \in \seur{\B}{w}$ we have $(\B, v) \nvDash \psi$. Thus $\BB \vDash \neg\di\psi$. So the formula $\di\psi$ 
separates $\AA$ from $\BB$ and since $\s(\di\psi) = \s(\psi) + 1 \le k$, $\sep_{k}$ holds.

\item The case in which the first move of the winning strategy of S is a right successor move
is similar to the case of left successor move. It suffices to switch the classes $\AA$ and $\BB$,
and replace $\di$ with $\bo$. 

\end{enumerate}

Now assume $\sep_{k}$ holds, and $\phi$ is the formula separating $\AA$ from $\BB$. 
We obtain a winning strategy of S 
for the game $\EF{k}(\AA, \BB)$ using 
$\phi$ as follows:
\begin{enumerate}[(a)]

\item If $\phi$ is a literal, S wins the game by making the corresponding literal move.

\item Assume that $\phi = \psi_1 \lor \psi_2$. Let $\AA_i := \{(\A, w) \in \AA \mid (\A, w) \vDash \psi_i\}$. Since $\AA \vDash \phi$ we have $\AA_1 \cup \AA_2 = \AA$. In addition, since $\BB \vDash \neg\phi$, we have $\BB \vDash \neg\psi_i$. Thus $\psi_i$ separates $\AA_i$ from $\BB$. Since $\s(\psi_1)+\s(\psi_2)+1=\s(\phi) \le k$, there are $k_1, k_2 \in \Nset$ such that $k_1+k_2+1 = k$ and $\s(\psi_i) \le k_i$.
By induction hypothesis S has winning strategies for the games $\EF{k_i}(\AA_i, \BB)$. Since $k \ge \s(\phi)
\ge 1$, S can start the game $\EF{k}(\AA, \BB)$ with a $\lor$-move choosing the numbers $k_1$ and $k_2$ and the sets $\AA_1$
and $\AA_2$. Then S wins the game by following the winning strategy for 
whichever position D chooses.

\item 
Assume that $\phi = \psi_1 \land \psi_2$. Let $\BB_1 := \{(\B, w) \in \BB \mid (\B, w) \nvDash \psi_1\}$ and $\BB_2 := \{(\B, w) \in \BB \mid (\B, w) \nvDash\psi_2$\}. Since $\BB \vDash \neg \phi$, we have $\BB_1 \cup \BB_2 = \BB$. In addition, since $\AA \vDash \phi$, we have $\AA \vDash \psi_1$ and $\AA\vDash \psi_2$. Thus $\psi_1$ separates $\AA$ from $\BB_1$ while $\psi_2$ separates $\AA$ from $\BB_2$. As in the previous case, there are $k_1, k_2 \in \Nset$ such that $k_1+k_2 = k$, $\size(\psi_1) \le k_1$ and $\size(\psi_2) \le k_2$. By induction hypothesis S has a winning strategy for the games $\EF{k}(\AA, \BB_1)$ and $\EF{k}(\AA, \BB_2)$. S wins the game $\EF{k}(\AA, \BB)$ by starting with a $\land$-move choosing the numbers $k_1$, and $k_2$ and the sets $\BB_1$ and $\BB_2$ and proceeding according to the winning strategies for the games $\EF{k}(\AA, \BB_1)$ and $\EF{k}(\AA, \BB_2)$.

\item Assume that $\phi = \di\psi$. Since $\AA \vDash \phi$, for every $(\A, w) \in \AA$ there is $(\A, v_w) \in \seur{\A}{w}$ such that $(\A, v_w) \vDash \psi$. 
We define the function $f : \AA \to \kaikki\AA$ by $f(\A, w) = (\A, v_w)$. Clearly $\di_f\AA \vDash \psi$. On the other hand $\BB \vDash \neg \phi$ so for each $(\B, w) \in \BB$ and each $(\B, v) \in \seur{\B}{w}$ we have $(\B, v) \nvDash \psi$. Therefore $\kaikki\BB \vDash \neg\psi$
and the formula $\psi$ separates $\di_f\AA$ from $\kaikki\BB$. Moreover,  $\s(\psi) = \s(\phi) - 1 \le k - 1$ so by
induction hypothesis S has a winning strategy for the game $\EF{k-1}(\di_f\AA, \kaikki\BB)$. Since $k \ge \s(\phi) \ge 1$, S can start the game 
$\EF{k}(\AA, \BB)$ with a $\di$-move choosing the function $f$. Then S wins the game by following the winning strategy for the game 
$\EF{k-1}(\di_f\AA, \kaikki\BB)$.

\item 
Assume finally that $\phi = \bo\psi$. Since $\AA \vDash \phi$, as in the previous case we obtain $\kaikki\AA \vDash \psi$. On the other hand, since $\BB \vDash \neg\phi$, for every $(\B, w) \in \BB$ there is $(\B, v_w) \in \seur{\B}{w}$ such that $(\B, v_w) \nvDash \psi$. We define the function $g: \BB \to \kaikki\BB$ by $g(\B, w) = (\B, v_w)$. Clearly $\di_g\BB \vDash \neg\psi$ so the formula $\psi$ separates the sets $\kaikki\AA$ and $\di_g\BB$. By induction hypothesis S has a winning strategy for the game $\EF{k-1}(\kaikki\AA, \di_g\BB)$. S wins the game $\EF{k}(\AA, \BB)$ by starting with a $\bo$-move choosing the function $g$ and proceeding according to the winning strategy of the game $\EF{k-1}(\kaikki\AA, \di_g\BB)$.
\end{enumerate}
\end{proof}

\begin{remark}
	In this form, the game $\EF{k}(\AA, \BB)$ tracks the size of the separating formula but with slight modifications it could track different things such as the number or nesting depth of specific operators. See e.g. the conference paper \cite{DBLP:conf/aiml/HellaV16} where the game counts propositional connectives and modal operators with two separate parameters. 
\end{remark}

Note that in Theorem~\ref{peruslause} we allow the set of proposition symbols $\Phi$ to be infinite. This is in contrast with other similar games, such as the bisimulation game and the $n$-bisimulation game. For an example let $\Phi = \{p_i \mid i \in \NN\}$ and $W = \{w\} \cup \{w_i \mid i \in \NN\}$. Furthermore let $(\A, w)$ be a pointed model, where $\dom(\A) = W$, $R^\A = \{(w, w_i) \mid i \in \NN\}$ and $V^\A(p_i) = \{w_j \mid j \ge i\}$ for each $i \in \NN$. Let $(\B, w)$ be the same model with the addition of a point $w_\NN$ in which all propositions are true. In other words $\dom(\B) = W \cup \{w_\NN\}$, $R^\B = R^\A \cup \{(b, w_\NN)\}$ and $V^\B(p_i) = V^\A(p_i) \cup \{w_\NN\}$ for each $i \in \NN$.

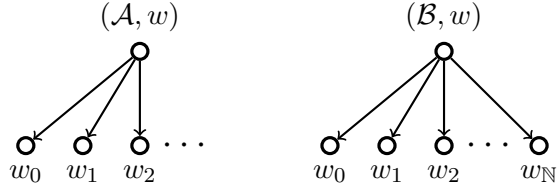
\begin{figure}[ht]
\centering
\begin{tikzpicture}
\node[circle, inner sep = 0pt, minimum size = 6pt, draw, very thick, label = above: {$(\A, w)$}] (a) at (-2,0) {};
\node[circle, inner sep = 0pt, minimum size = 6pt, draw, very thick, label = above: {$(\B, w)$}] (b) at (2,0) {};

\node[circle, inner sep = 0pt, minimum size = 6pt, draw, very thick, label = below: {$w_0$}] (a1) at (-3.5,-1.25) {};
\node[circle, inner sep = 0pt, minimum size = 6pt, draw, very thick, label = below: {$w_1$}] (a2) at (-2.75,-1.25) {};
\node[circle, inner sep = 0pt, minimum size = 6pt, draw, very thick, label = below: {$w_2$}] (a3) at (-2,-1.25) {};

\draw[->, thick] (a) to (a1);
\draw[->, thick] (a) to (a2);
\draw[->, thick] (a) to (a3);

\node[font = \Large] (apisteet) at (-1.375,-1.25) {$\cdots$};

\node[circle, inner sep = 0pt, minimum size = 6pt, draw, very thick, label = below: {$w_0$}] (b1) at (0.5,-1.25) {};
\node[circle, inner sep = 0pt, minimum size = 6pt, draw, very thick, label = below: {$w_1$}] (b2) at (1.25,-1.25) {};
\node[circle, inner sep = 0pt, minimum size = 6pt, draw, very thick, label = below: {$w_2$}] (b3) at (2,-1.25) {};
\node[circle, inner sep = 0pt, minimum size = 6pt, draw, very thick, label = below: {$w_\NN$}] (b4) at (3.25,-1.25) {};

\draw[->, thick] (b) to (b1);
\draw[->, thick] (b) to (b2);
\draw[->, thick] (b) to (b3);
\draw[->, thick] (b) to (b4);

\node[font = \Large] (apisteet) at (2.625,-1.25) {$\cdots$};
\end{tikzpicture}
\caption{The pointed models $(\A, w)$ and $(\B, w)$.}
\label{esimkuva}
\end{figure}


We see that by moving to $w_\NN$, S wins the ($n$-)bisimulation game between the models $(\A, w)$ and $(\B, w)$, even though the models satisfy exactly the same $\ML$-formulas. 

We prove next that $k$-bisimilarity implies that D has winning strategy in the 
formula size game with resource parameter $k$. 
This simple observation is used in the next section, when we apply the game $\EF{k}$ for proving a succinctness result for $\FO$ over $\ML$.

\begin{theorem}\label{bisim}
Let $\AA$ and $\BB$ be sets of pointed models and let $k \in \Nset$. If there are $(k-1)$-bisimilar pointed models $(\A, w) \in \AA$ and $(\B, v) \in \BB$, 
then D has a winning strategy for the game $\EF{k}(\AA, \BB)$.
\end{theorem}
\begin{proof}
The proof proceeds by induction on the number $k \in \Nset$. If $k = 1$ and $(\A, w) \in \AA$ and $(\B, v) \in \BB$ are 
$0$-bisimilar and thus satisfy the same literals. Thus there is no literal $\phi \in \ML$ that separates the sets $\AA$ and $\BB$. Thus any literal move by S leads to D winning. In addition, any non-literal move leads to a following position with $k = 0$ so D wins the game $\EF{1}(\AA, \BB)$.

Assume that $k > 1$ and $(\A, w) \in \AA$ and $(\B, v) \in \BB$ are $(k-1)$-bisimilar. We
consider the cases of the first move of S in the game $\EF{k}(\AA, \BB)$.

If S makes a literal move, D will win as in the basic step.

If S starts with a $\lor$-move choosing the numbers $k_1$ and $k_2$ and the sets $\AA_1$ and $\AA_2$, then since $\AA_1 \cup \AA_2
= \AA$, D can choose the next position $(k_i, \AA_i, \BB)$, in such a way that $(\A, w) \in \AA_i$. Then we have
$k_i < k$ so by induction hypothesis D has a winning strategy for the game $\EF{k_i}(\AA_i, \BB)$. The case of a $\land$-move is similar.

If S starts with a $\di$-move choosing a function $f : \AA \to \kaikki\AA$, then since $(\A, w)$ and $(\B, v)$ are $(k-1)$-bisimilar, there is a pointed
model $(\B, v') \in \seur{\B}{v}$ that is $(k-2)$-bisimilar with the pointed model $f(\A, w)$. By induction hypothesis D has a winning
strategy in $\EF{k-1}(\di_f\AA, \kaikki\BB)$. The case of a $\bo$-move is similar.
\end{proof}

\section{Succinctness of FO over ML}\label{nelonen}

In this section, we illustrate the use of the formula size game $\EF{k}$ by proving
a non-elementary succinctness gap between bisimulation invariant first-order logic
and modal logic. We also show that this gap is already present between a limited 2-dimensional modal logic  $\ML^2$ and basic modal logic.

A similar gap between $\FO$ and linear temporal logic $\LTL$ has already been established in the literature. In his PhD thesis \cite{stockmeyer}, Stockmeyer proved that the satisfiability problem of $\FO$ over words is of non-elementary complexity. He reduced the problem of nonemptiness of star-free regular expressions to this satisfiability problem. Etessami and Wilke pointed out in \cite{fo2} that careful examination of Stockmeyer's proof yields $\FO$ sentences with size $\ordo(n)$ such that the minimal words satisfying these sentences have length non-elementary in $n^3$. Since all satisfiable formulas of $\LTL$ have a satisfying model at most exponential in the size of the formula, a non-elementary succinctness gap between $\FO$ and $\LTL$ is obtained.

\subsection{A property of pointed models}\label{sec-property}

For the remainder of this section we consider only the case where the set $\Phi$ of propositional symbols is empty. This makes all points in Kripke-models propositionally
equivalent so the only formulas available for the win condition of S in the game $\EF{k}$ are
$\bot$ and $\top$. Thus S can only win with a literal move from position $(k,\AA, \BB)$ if either $\AA = \emptyset$ and $\BB \neq \emptyset$, or 
$\AA \neq \emptyset$ and $\BB = \emptyset$.

We will use the following two classes in our application of the formula size game $\EF{k}$:
\begin{itemize}
\item $\AA_n$ is the class of all pointed models $(\A, w)$ such that for all
 $(\A, u), (\A, v) \in \seur{\A}{w}$, the models $(\A, u)$ and $(\A, v)$ are $n$-bisimilar.
\item $\BB_n$ is the complement of $\AA_n$.
\end{itemize}

\begin{lemma}\label{FO}
For each $n \in \Nset$ there is a formula $\phi_n(x) \in \FO$ that separates the classes $\AA_n$ and $\BB_n$ such that the size of $\phi_n(x)$ is linear with 
respect to $n$, i.e., $\s(\phi_n) = \ordo(n)$.
\end{lemma}
\begin{proof}
We first define formulas $\psi_n(x,y) \in \FO$ such that 
$(\M, u)\bisim_n (\M, v)$ if and only if $\M \vDash \psi_n[u/x, v/y]$.
The formulas $\psi_n(x,y)$ are defined recursively as follows:
\begin{align*}
\psi_1(x,y) := &\exists s R(x,s) \leftrightarrow \exists t R(y,t) \\
\psi_{n+1}(x,y) := &\forall s \exists t\Big(\big(R(x,s) \rightarrow R(y,t)\big) \land \big(R(y,s) \rightarrow R(x,t)\big)  \\& \ \qquad \land \big(R(x,s) \lor R(y,s) \rightarrow \psi_n(s,t)\big)\Big) 
\end{align*}
Clearly these formulas express $n$-bisimilarity as intended. When we interpret the equivalence and implications as shorthand in the standard way,
we get the sizes $\s(\psi_1) = 11$ and $\s(\psi_{n+1}) = \s(\psi_n) + 14$. Thus $\s(\psi_n) = 14n-3$.

Now we can define the formulas $\phi_n$:
\[
\phi_n(x) := \forall y \forall z (R(x,y) \land R(x,z) \rightarrow \psi_n(y,z)).
\]
Clearly for every $(\A, w) \in \AA_n$ we have $\A \vDash \phi_n[w/x]$ and for every $(\B, v) \in \BB_n$ we have $\B \vDash \neg \phi_n[w/x]$ so the formula
$\phi_n$ separates the classes $\AA_n$ and $\BB_n$. Furthermore, $\s(\phi_n) = \s(\psi_n) + 6 = 14n+3$ so the size of $\phi_n$ is linear\footnote{{\bf Acknowledgement.} We are grateful to Martin L{\"u}ck for coming up with a linear size formula $\psi_n(x,y)$ to replace our previous one that was of exponential size.} with respect to $n$.
\end{proof}

\begin{lemma}
For each $n \in \Nset$, the formula $\phi_n$ is $(n+1)$-bisimulation invariant.
\end{lemma}
\begin{proof}
Let $(\A, w)$ and $(\B, v)$ be $(n+1)$-bisimilar pointed models. Assume that $\A \vDash \phi_n[w/x]$. If $(\B, v_1), (\B, v_2) \in \seur{\B}{v}$, by $(n+1)$-bisimilarity there are 
$(\A, w_1), (\A, w_2) \in \seur{\A}{w}$ such that $(\A, w_1) \bisim_n (\B, v_1)$ and $(\A, w_2) \bisim_n (\B, v_2)$. Since $\A \vDash \phi_n[w/x]$, we have 
$(\B, v_1) \bisim_n (\A, w_1) \bisim_n (\A, w_2) \bisim_n (\B, v_2)$ so $\B \vDash \psi_n[v_1/x, v_2/y]$. Thus, we see that $\B \vDash \phi_n[v/x]$.
\end{proof}

It follows now from van Benthem's characterization theorem that each $\phi_n$ is equivalent
to some $\ML$-formula. Thus, we get the following corollary.

\begin{corollary}\label{Rosen}
For each $n \in \Nset$, there is a formula $\vartheta_n \in \ML$ that separates the classes $\AA_n$ and $\BB_n$.
\end{corollary}

\subsection{Set theoretic construction of pointed models}\label{construction}

We have shown that the classes $\AA_n$ and $\BB_n$ can be separated both in $\ML$ and in $\FO$. Furthermore the size of the FO-formula is linear with respect to
$n$. It only remains to ask: what is the size of the smallest $\ML$-formula that separates the classes $\AA_n$ and $\BB_n$? To answer this we will need suitable subsets of $\AA_n$ and $\BB_n$ to play the formula size game
on.

\begin{definition}
Let $n \in \Nset$. \emph{The finite levels of the cumulative hierarchy} are defined recursively as follows:
\begin{align*}
\mathsf{V}_0 &= \emptyset \\
\mathsf{V}_{n+1} &= \powerset(\mathsf{V}_n)
\end{align*}
\end{definition}

For every $n \in \Nset$, $\mathsf{V}_n$ is a \emph{transitive set}, i.e., for every $a \in \mathsf{V}_n$ and every $b \in a$ it holds that $b \in \mathsf{V}_n$. Thus it is reasonable to define a model $\F_n = (\mathsf{V}_n, R_n)$, where for all $a, b \in \mathsf{V}_n$ it holds that $(a, b) \in R_n \Leftrightarrow b \in a$.

For every point $a \in \mathsf{V}_n$ we denote by $(\M_a, a)$ the pointed model, where $\M_a$ is the submodel of $\F_n$ generated by the point 
$a$. 

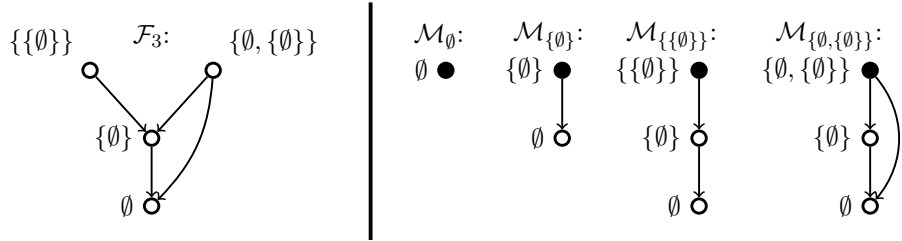
\begin{figure}[ht]
\centering
\scalebox{0.9}{
\begin{tikzpicture}
\draw[ultra thick] (1.7,3.5) to (1.7,0);

\node at (-1.5, 3) {$\F_3$:};
\node[circle, inner sep = 0pt, minimum size = 6pt, draw, very thick, label = above right : {$\{\emptyset, \{\emptyset\}\}$}] (1) at (-0.6, 2.5) {};
\node[circle, inner sep = 0pt, minimum size = 6pt, draw, very thick, label = left : $\{\emptyset\}$] (2) at (-1.5, 1.5) {};
\node[circle, inner sep = 0pt, minimum size = 6pt, draw, very thick, label = left : $\emptyset$] (3) at (-1.5, 0.5) {};
\node[circle, inner sep = 0pt, minimum size = 6pt, draw, very thick, label = above left : $\{\{\emptyset\}\}$] (4) at (-2.4, 2.5) {};

\draw[->, thick] (1) to (2);
\draw[->, thick] (2) to (3);
\draw[->, thick] (1) to [bend left = 20] (3);
\draw[->, thick] (4) to (2);

\node at (2.7, 3) {$\M_\emptyset$:};
\node[circle, fill, inner sep = 0pt, minimum size = 6pt, draw, very thick, label = left : $\emptyset$] (2) at (2.8, 2.5) {};

\node at (4.3, 3) {$\M_{\{\emptyset\}}$:};
\node[circle, fill, inner sep = 0pt, minimum size = 6pt, draw, very thick, label = left : $\{\emptyset\}$] (1) at (4.5, 2.5) {};
\node[circle, inner sep = 0pt, minimum size = 6pt, draw, very thick, label = left : $\emptyset$] (2) at (4.5, 1.5) {};

\draw[->, thick] (1) to (2);

\node at (6.1, 3) {$\M_{\{\{\emptyset\}\}}$:};
\node[circle, fill, inner sep = 0pt, minimum size = 6pt, draw, very thick, label = left : $\{\{\emptyset\}\}$] (3) at (6.5, 2.5) {};
\node[circle, inner sep = 0pt, minimum size = 6pt, draw, very thick, label = left : $\{\emptyset\}$] (1) at (6.5, 1.5) {};
\node[circle, inner sep = 0pt, minimum size = 6pt, draw, very thick, label = left : $\emptyset$] (2) at (6.5, 0.5) {};

\draw[->, thick] (1) to (2);
\draw[->, thick] (3) to (1);

\node at (8.4, 3) {$\M_{\{\emptyset, \{\emptyset\}\}}$:};
\node[circle, fill, inner sep = 0pt, minimum size = 6pt, draw, very thick, label = left : {$\{\emptyset, \{\emptyset\}\}$}] (3) at (9, 2.5) {};
\node[circle, inner sep = 0pt, minimum size = 6pt, draw, very thick, label = left : $\{\emptyset\}$] (1) at (9, 1.5) {};
\node[circle, inner sep = 0pt, minimum size = 6pt, draw, very thick, label = left : $\emptyset$] (2) at (9, 0.5) {};

\draw[->, thick] (1) to (2);
\draw[->, thick] (3) to (1);
\draw[->, thick] (3) to [bend left = 40] (2);

\end{tikzpicture}}
\caption{The model $\F_3$ and its generated submodels}
\end{figure}

\begin{lemma}\label{notbisim}
Let $n \in \Nset$ and $a, b \in \mathsf{V}_{n+1}$. If $a \neq b$, then $(\M_a, a)\not\bisim_n(\M_b, b)$. 
\end{lemma}
\begin{proof}
We prove the claim by induction on $n$. The basic step $n = 0$ is trivial since $\mathsf{V}_1$ only has one element. For the induction step, assume that 
$a, b \in \mathsf{V}_{n+1}$ and $a \neq b$. Assume further for contradiction that 
$(\M_a, a)\bisim_n (\M_b, b)$. Since $a \neq b$, by symmetry we can 
assume that there is $x \in a$ such that $x \notin b$. By $n$-bisimilarity there is $y \in b$ such that $(\M_x, x)$ and $(\M_y, y)$ are $(n-1)$-bisimilar. Since 
$x \in a \in \mathsf{V}_{n+1}$ and $y \in b \in \mathsf{V}_{n+1}$, we have $x ,y \in \mathsf{V}_n$. By induction hypothesis we obtain $x = y$. This is a contradiction, since $x \notin b$ 
and $y \in b$.
\end{proof}

If $\AA$ is a set of pointed models, the pointed model $\liitos \AA$ is formed by taking all the pointed models of $\AA$ and connecting a new root to their distinguished points as illustrated in Figure \ref{liitoskuva}. 
To make sure that $(\liitos \AA,v)$ is bisimilar with $(\A,v)$ for any $(\A,v)\in\liitos \AA$,
we require that the models in $\AA$ are compatible in possible intersections. The precise definition is the following.

Let $\AA$ be a set of pointed models such that for all $(\A, v), (\A', v') \in \AA$ it holds that $R^\A \upharpoonright (\dom(\A) \cap \dom(\A')) = R^{\A'} \upharpoonright (\dom(\A) \cap \dom(\A'))$ and let $w \notin \dom(\A)$ for all $(\A, v) \in \AA$. We use the notation $\liitos\AA := (\M, w)$, where 
\begin{align*}
&\dom(\M) = \{w\} \cup \bigcup\{\dom(\A) \mid (\A, v) \in \AA\}, \text{ and }\\
 &R^\M = \{(w,v) \mid (\A, v) \in \AA \} \cup \bigcup\{R^\A \mid (\A, v) \in \AA\}.
\end{align*}


\begin{figure}[ht]
\centering
\begin{tikzpicture}
\node[circle, inner sep = 0pt, minimum size = 6pt, draw, very thick, label = left: {$w$}] (w) at (0,0) {};
\node (A) at (0,0.5) {$\liitos\AA$};
\node[circle, inner sep = 0pt, minimum size = 6pt, draw, very thick] (1) at (-1.875,-1.25) {};
\node[circle, inner sep = 0pt, minimum size = 6pt, draw, very thick] (2) at (-0.625,-1.25) {};
\node[circle, inner sep = 0pt, minimum size = 6pt, draw, very thick] (3) at (1.875,-1.25) {};
\node[font = \Large] (pisteet) at (0.625,-1.25) {$\cdots$};

\draw[->, thick] (w) to (1);
\draw[->, thick] (w) to (2);
\draw[->, thick] (w) to (3);

\draw (1) -- ++(-0.4,-1) -- ++(0.8,0) -- (1);
\draw (2) -- ++(-0.4,-1) -- ++(0.8,0) -- (2);
\draw (3) -- ++(-0.4,-1) -- ++(0.8,0) -- (3);

\draw[decorate, decoration = {brace, mirror}] (-2.30,-2.35) -- (2.30,-2.35);
\node (A) at (0, -2.65) {$\AA$};
\end{tikzpicture}
\caption{The pointed model $\liitos\AA$}
\label{liitoskuva}
\end{figure}


For each $n \in \NN$ we define the following sets of pointed models:
\begin{align*}
\CC_n &:= \{ \liitos \{(\M_a, a)\} \mid a \in \mathsf{V}_{n+1} \} \\
\DD_n &:= \{ \liitos \{ (\M_a, a), (\M_b, b) \} \mid a, b \in \mathsf{V}_{n+1}, a \neq b \}.
\end{align*}

In other words the pointed models in $\CC_n$ have a single successor from level $n+1$ of the cumulative hierarchy, whereas the pointed models in $\DD_n$ have two different successors from the same set. Therefore clearly $\CC_n \subseteq \AA_n$ and by Lemma~\ref{notbisim} also $\DD_n \subseteq \BB_n$. In the next subsection we will use these sets in the formula size game.

It is well known that the cardinality of $\mathsf{V}_n$ is the exponential tower of $n-1$. 
Thus, the cardinality of $\CC_n$ is $\tower(n)$.

\begin{lemma}\label{tower}
If $n \in \Nset$, we have $|\CC_n| = |\mathsf{V}_{n+1}| = \tower(n)$. \qed
\end{lemma}

\subsection{Graph colorings and winning strategies in $\EF{k}$}\label{graphs}

Our aim is to prove that any $\ML$-formula $\theta_n$ separating the
sets $\CC_n$ and $\DD_n$ is of size at least $\tower(n-1)$. To do this,
we make use of a surprising connection between the chromatic numbers
of certain graphs related to pairs of the form $(\VV,\EE)$, where $\VV\subseteq\CC_n$
and $\EE\subseteq\DD_n$, and existence
of a winning strategy for D in the game $\EF{k}(\VV,\EE)$.

Let $n \in \Nset$, $\emptyset \neq \VV \subseteq \CC_n$ and $\EE \subseteq \DD_n$. Then $\G (\VV, \EE)$ denotes the graph $(V, E)$, where
\begin{align*}
V &= \bo\VV \text{ and }\\
 E &= \{ ((\M, w), (\M', w')) \in V \times V \mid \liitos \{(\M, w), (\M', w')\} \in \EE\}.
\end{align*}
 That is, since models on the left all have exactly one successor, and ones on the right have exactly two successors from the same basic set, we can take the graph where these successors are nodes and the pairs on the right define the edges. Note that a pair on the right only produces an edge if both elements of the pair are present on the left.

\begin{definition}
Let $\G = (V, E)$ be a graph and let $C$ be a set. A function $\chi : V \to C$ is a \emph{coloring} of the graph $\G$ if for all $u, v \in V$ it holds that 
if $(u, v) \in E$, then $\chi(u) \neq \chi(v)$. If the set $C$ has $k$ elements, then $\chi$ is called a \emph{$k$-coloring} of $\G$.

The \emph{chromatic number} of $\G$, denoted by $\chi(\G)$, is the smallest number $k \in \Nset$ for which there is a $k$-coloring of $\G$.
\end{definition}

When playing the formula size game $\EF{k}(\VV, \EE)$, connective moves correspond with dividing either the vertex set or the edge set of the graph $\G(\VV, \EE)$ into two parts, forming two new graphs. In the next lemma we get simple arithmetic estimates for the behaviour of chromatic numbers in such divisions. In the case of a vertex set split, if the two new graphs are colored with separate colors, combining these colorings yields a coloring of the whole graph. For an edge split, the full graph is colored with pairs of colors given by the two new colorings. If two vertices are adjacent in the full graph, at least one of the new colorings will color them with a different color and the pairs of colors will be different.

\begin{lemma}\label{väritys}
Let $\G = (V, E)$ be a graph. 
\begin{enumerate}
\item Let $V_1, V_2 \subseteq V$ be nonempty such that $V_1 \cup V_2 = V$ and let $\G_1 = (V_1, E \upharpoonright V_1)$ and
$\G_2 = (V_2, E \upharpoonright V_2)$. Then we have  $\chi (\G) \le \chi (\G_1) + \chi (\G_2)$.
\item Let $E_1, E_2 \subseteq E$ such that $E_1 \cup E_2 = E$ and let  $\G_1 = (V, E_1)$ and \mbox{$\G_2 = (V, E_2)$}. Then $\chi (\G) \le \chi (\G_1)\chi (\G_2)$.
\end{enumerate}
\end{lemma}
\begin{proof}
\begin{enumerate}
\item 
Let $V_1$, $V_2$, $\G_1$ and $\G_2$ be as in the claim and let $k_1 = \chi(\G_1)$ and $k_2 = \chi(\G_2)$. Let $\chi_1 : V_1 \to \{1, \dots, k_1\}$ 
be a $k_1$-coloring of the graph $\G_1$ and let $\chi_2 : V_2 \to \{k_1+1, \dots, k_1+k_2 \}$ be a $k_2$-coloring of the graph $\G_2$. Then it is straightforward to show that
$\chi = \chi_1 \cup (\chi_2 \upharpoonright (V_2 \setminus V_1))$ is a $k_1+k_2$-coloring of the graph $\G$, whence $\chi (\G) \le k_1+k_2 =\chi (\G_1) + \chi (\G_2)$.

\item 
Let $\chi_1 : V \to \{1, \dots, k_1\}$ and $\chi_2 : V \to \{1, \dots, k_2\}$ be colorings of the graphs $\G_1$ and $\G_2$, respectively. Then it is easy to verify that the map 
$\chi : V \to \{1, \dots, k_1\} \times \{1, \dots, k_2\}$ defined by $\chi (v) = (\chi_1(v), \chi_2(v))$ is a coloring of $\G$. Thus we obtain $\chi (\G) \le |\{1, \dots, k_1\} \times 
\{1, \dots, k_2\}| = \chi (\G_1) \chi (\G_2)$.
\end{enumerate}
\end{proof}

For the condition D maintains to win the game, we use the logarithm of the chromatic number of $\G(\VV, \EE)$ as it behaves nicely with both kinds of splittings. Note that to achieve non-elementary formula size, it suffices to consider the number of binary connectives required before any modal moves can be made.

\begin{lemma}\label{winstrat}
Assume $\emptyset \neq \VV \subseteq \CC_n$ and $\EE \subseteq \DD_n$ for some $n \in \Nset$ and let $k \in \Nset$. If $k \leq \log ( \chi (\G(\VV, \EE)))$, then D has a winning strategy in the game $\EF{k}(\VV, \EE)$.
\end{lemma}
\begin{proof}
Let $n,k \in \Nset$ and assume that $\emptyset \neq \VV \subseteq \CC_n$, $\EE \subseteq \DD_n$ and $k \leq \log (\chi (\G (\VV, \EE)))$.
We prove the claim by induction on $k$.

If $k = 0$, then D wins the game.

If $k = 1$, any non-literal move of S leads to D winning. Since $\VV, \EE \neq \emptyset$ and all models are propositionally equivalent, D will also win if S makes a literal move. 

Assume then that $k > 1$. If S starts the game with a literal move, then D wins as described above.

Assume that S begins the game with a $\di$- or $\bo$-move. Since $\chi (\G (\VV, \EE)) \ge 2$, there are pointed models $(\M, w), (\M', w') \in V$ such that $((\M, w), (\M', w')) \in E$.
Thus $\liitos \{(\M, w)\}$, $\liitos \{(\M', w')\} \in \VV$ and $\liitos \{(\M, w), (\M',w')\} \in \EE$.  In the 
following position $(k-1, \VV', \EE')$ it holds that $(\M, w) \in \VV' \cap \EE'$ or $(\M', w') \in \VV' \cap \EE'$. Thus the same pointed model is present on both
sides of the game and by Theorem~\ref{bisim}, D has a winning strategy for the game $\EF{k-1}(\VV', \EE')$.

Assume that S begins the game with a $\lor$-move choosing the numbers $k_1, k_2 \in \Nset$ and the sets $\VV_1, \VV_2 \subseteq \VV$. Consider
the graphs $\G (\VV, \EE) = (V, E)$ and $\G(\VV_i, \EE) = (V_i, E_i)$. Since $\VV_1 \cup \VV_2 = \VV$, we have $V_1 \cup
V_2 = V$. In addition, by the definition of the graphs $\G(\VV, \EE)$ and $\G(\VV_i, \EE)$ we see that $E_i = E \upharpoonright V_i$. 
Thus by Lemma~\ref{väritys}, we obtain $\chi(\G(\VV, \EE)) \le \chi(\G(\VV_1, \EE)) + \chi(\G(\VV_2, \EE))$. It must hold that $k_1 \leq \log(\chi(\G(\VV_1, \EE)))$ or
$k_2 \leq \log(\chi(\G(\VV_2, \EE)))$, since otherwise we would have
\begin{align*}
k &\leq \log(\chi(\G(\VV, \EE))) \le \log(\chi(\G(\VV_1, \EE)) + \chi(\G(\VV_2, \EE)))  \\ 
& \le \log(\chi(\G(\VV_1, \EE))) + \log(\chi(\G(\VV_2, \EE))) + 1 < k_1 + k_2 + 1 = k.
\end{align*}
Thus D can choose the next position of the game, $(k_i, \VV_i, \EE)$, in such a way that 
$k_i \leq \log(\chi(\G(\VV_i, \EE)))$. By induction hypothesis
D has a winning strategy in the game $\EF{k_i}(\VV_i, \EE)$. 

Assume then that S begins the game with a $\land$-move choosing the numbers $k_1, k_2 \in \Nset$ and the sets $\EE_1, \EE_2 \subseteq \EE$. Consider
now the graphs $\G(\VV, \EE)=(V, E)$ and $\G(\VV, \EE_i)=(V_i, E_i)$. Clearly $V_1 = V_2 = V$ and since 
$\EE_1 \cup \EE_2 = \EE$, we have $E_1 \cup E_2 = E$. Thus by Lemma~\ref{väritys}, we obtain  $\chi(\G(\VV, \EE)) \le \chi(\G(\VV, \EE_1))\chi(\G(\VV, \EE_2))$.
It must hold that $k_1 \leq \log(\chi(\G(\VV, \EE_1)))$ or $k_2 \leq \log(\chi(\G(\VV, \EE_2)))$, since otherwise we would have
\begin{align*}
k &\leq \log(\chi(\G(\VV, \EE))) \le \log(\chi(\G(\VV, \EE_1))\chi(\G(\VV, \EE_2))) \\
& = \log(\chi(\G(\VV, \EE_1))) + \log(\chi(\G(\VV, \EE_2))) < k_1+k_2+1 = k.
\end{align*}
Thus D can again choose the next position of the game, $(k_i, \VV, \EE_i)$, in such a way that $k_i \leq \log(\chi(\G(\VV, \EE_i)))$. By induction hypothesis
D has a winning strategy in the game $\EF{k_i}(\VV, \EE_i)$.
\end{proof}

\begin{theorem}\label{ML}
Let $n \in \Nset$. If a formula $\theta_n \in \ML$ separates $\AA_n$ from $\BB_n$, then $\s(\theta_n) > \tower(n-1)$.
\end{theorem}
\begin{proof}
Assume that a formula $\theta_n \in \ML$ separates $\AA_n$ from $\BB_n$. As observed in the end of Subsection \ref{construction}, it holds that $\CC_n \subseteq \AA_n$ and $\DD_n \subseteq \BB_n$. Therefore $\theta_n$ also separates the sets $\CC_n$ and $\DD_n$.

Assume for contradiction that $\s(\theta_n) \leq \tower(n-1)$. By Theorem~\ref{peruslause}, S has a winning strategy in the game $\EF{k}(\CC_n, \DD_n)$ for $k = \s(\theta_n)$.

On the other hand, by Lemma~\ref{tower}, we have $|\CC_n| = \tower(n)$ and the set $\DD_n$ consists of all the pointed models $\liitos\{(\M, w), (\M', w')\}$, where $\liitos\{(\M, w)\},
\liitos\{(\M', w')\} \in \CC_n$, $(\M, w) \neq (\M', w')$. Thus the graph $\G(\CC_n, \DD_n)$ is isomorphic with the complete graph $K_{\tower(n)}$. Therefore we obtain
\[
\chi(\G(\CC_n, \DD_n)) = \chi(K_{\tower(n)}) = \tower(n).
\]
By the assumption, $k \leq \tower(n-1) = \log(\tower(n)) = \log(\chi(\G(\CC_n, \DD_n)))$, so by Lemma~\ref{winstrat}, D also has a winning strategy in the game 
$\EF{k}(\CC_n, \DD_n)$, which is a contradiction.
\end{proof}

We now have everything we need for proving the non-elementary succinctness of $\FO$ 
over $\ML$. By Lemma~\ref{FO}, for each $n \in \Nset$ there is a formula $\phi_n(x) \in \FO$ such that $\phi_n$ separates the classes $\AA_n$ and $\BB_n$ with $s(\phi)=\ordo(n)$. 
On the other hand by Corollary \ref{Rosen}, there is an equivalent formula $\vartheta_n \in \ML$, but by Theorem~\ref{ML} the size of $\vartheta_n$ 
must be at least $\tower(n-1)$. 
So the property of all successors of a pointed model being $n$-bisimilar with each other can be expressed in $\FO$ with a formula of linear size, but in $\ML$ expressing it requires a formula of non-elementary size. 

\begin{corollary}\label{succinct}
Bisimulation invariant $\FO$ is non-elementarily more succinct than $\ML$.
\end{corollary}

\begin{remark}
It is well known that the DAG-size of any formula $\phi$ is greater than or equal to the logarithm of the size of $\phi$. Thus if $\theta_n$ is a formula as in Theorem~\ref{ML}, the DAG-size of $\theta_n$ must be at least $\tower(n-2)$. Consequently the result of Corollary~\ref{succinct} also holds for DAG-size.
\end{remark}

\subsection{Succinctness of $2$-dimensional modal logic}

Our proof for the non-elementary succinctness gap between bisimulation invariant $\FO$ and $\ML$
is based on the fact that $n$-bisimilarity of two points $u,v\in W$ of a Kripke-model 
$\M=(W,R)$ is definable by a linear $\FO$-formula $\psi_n(x,y)$ (see the proof of Lemma~\ref{FO}). 
We will now show that the property $(\M, u)\bisim_n (\M, v)$ is succinctly expressible also in 
\emph{$2$-dimensional modal logic}. 


The idea in $2$-dimensional modal logic is that the truth of formulas is evaluated on 
pairs $(u,v)$ of points of Kripke-models instead of single points. 
We refer to the book \cite{MarxV97} of 
Marx and Venema and the series of papers \cite{Shehtman}, \cite{Gabbay2000-GABPOM-2}, \cite{Gabbay2002-GABPOM} of Gabbay and Shehtman for a detailed exposition on $2$-dimensional and multi-dimensional
modal logics. For our purposes it suffices to consider the logic Gabbay and Shehtman call $\mathbf{K}^2$. For consistency of notation in this paper we call the logic $\ML^2$ and introduce it only semantically.

A Kripke-model $\T$ for $\ML^2$ consists of a set $W$ of points, two binary accessibility 
relations $R_1$ and $R_2$, and a valuation $V : \Phi \to \powerset(W)$. Correspondingly,
$\ML^2$ has two modal operators $\di_1,\di_2$ and their duals $\bo_1,\bo_2$.
The semantics of these operators is defined as follows:
\begin{itemize}
\item $(\T, (u,v)) \vDash \di_1\phi \Leftrightarrow$ there is $u'\in W$ such that 
$uR_1u'$ and $(\T, (u',v)) \vDash \phi$,
\item $(\T, (u,v)) \vDash \di_2\phi \Leftrightarrow$ there is $v'\in W$ such that 
$vR_2v'$ and $(\T, (u,v')) \vDash \phi$,
\item $(\T, (u,v)) \vDash \bo_1\phi \Leftrightarrow$ for all $u'\in W$, if 
$uR_1u'$, then $(\T, (u',v)) \vDash \phi$,
\item $(\T, (u,v)) \vDash \bo_2\phi \Leftrightarrow$ for all $v'\in W$, if  
$vR_2v'$, then $(\T, (u,v')) \vDash \phi$.
\end{itemize}

Any pointed Kripke-model $(\M,w)=((W,R,V),w)$ can be interpreted as the $2$-dimensional pointed model
$(\M_2,(w,w))$, where $\M_2=(W,R,R,V)$. This gives us a meaningful way of defining properties of
pointed models $(\M,w)$ by formulas of $\ML^2$. In particular, we say that a formula 
$\varphi\in\ML^2$ separates two classes $\AA$ and $\BB$ of pointed models if 
for all $(\M,w)\in \AA$, $ (\M_2,(w,w))\vDash \varphi$ and for all $(\M,w)\in \BB$, 
$(\M_2,(w,w))\nvDash\varphi$.

The \emph{size} $\s(\varphi)$ of a formula $\varphi\in\ML^2$ is defined in the same
way as for formulas of $\ML$; see Definition~\ref{mlsize}.
In other words, $\s(\varphi)$ is the total number of modal operators, binary connectives and literals
occurring in $\varphi$.

Observe now that two pointed models $(\M,u)$ and $(\M,v)$ with no propositional symbols are $1$-bisimilar if and only if
$(\M_2,(u,v))\vDash\rho_1$, where $\rho_1:=\di_1\top\leftrightarrow\di_2\top$. 
Furthermore if $\rho_n\in\ML_2$ defines the class of all $2$-dimensional pointed models 
$(\M_2,(u,v))$ such that $(\M,u)\bisim_n(\M,v)$,
then $\rho_{n+1}:=\bo_1\di_2\rho_n\land\bo_2\di_1\rho_n$ defines the class of all 
$(\M_2,(u,v))$ such that $(\M,u)\bisim_{n+1}(\M,v)$.

\begin{lemma}\label{2-dim}
For each $n \in \Nset$ there is a formula $\zeta_n \in \ML^2$ that separates the classes $\AA_n$ and $\BB_n$ such that the size of $\zeta_n$ is exponential with 
respect to $n$, i.e., $\s(\zeta_n) = \ordo(2^n)$.
\end{lemma}
\begin{proof}
Let $\zeta_n$ be the formula $\bo_1\bo_2\rho_n$. Then $(\M_2,(w,w))\vDash\zeta_n$ if and only if 
$(\M,u)$ and $(\M,v)$
are $n$-bisimilar for all $(\M, u), (\M, v) \in \seur{\M}{w}$, whence $\zeta_n$ separates $\AA_n$ from its complement $\BB_n$.
An easy calculation shows that the size of $\zeta_n$ is $2^{n+4}-3$. 
\end{proof}

By Theorem~\ref{Rosen}, for each $n\in\Nset$ there is a formula $\vartheta_n\in\ML$
that is equivalent with $\zeta_n$. On the other hand, by Theorem~\ref{ML} the size of $\vartheta_n$ 
is at least $\tower(n-1)$. Thus, we obtain the non-elementary succinctness gap already between $\ML^2$
and $\ML$.

\begin{corollary}
The $2$-dimensional modal logic $\ML^2$ is non-elementarily more succinct than $\ML$.
\end{corollary}

\section{The formula size game for $\L_\mu$}\label{mugamesection}

To define a formula size game similar to the one of $\ML$ for $\L_\mu$, we will need some additional notation and concepts, since $\L_\mu$ is significantly more complex than $\ML$. 

Let $(\set, \rel)$ be a tree and let $s, t \in \set$. We say that $s$ is \emph{above} $t$ if there is an $E$-path from $s$ to $t$. We say that $s$ is \emph{below} $t$ if $t$ is above $s$. A triple $(\set, \rel, \back)$ is a \emph{tree with back edges} if $(\set, \rel)$ is a tree and $s$ is below $t$ for every $(s, t) \in \back$.

We define for each formula $\phi \in \L_\mu$ its \emph{syntax tree with back edges}, $T_\phi = (V_\phi, E_\phi, B_\phi, \lab_\phi)$ as follows. The set $V_\phi$ consists of occurrences of subformulas of $\phi$ and the relation $E_\phi$ is the subformula relation between those occurrences. Additionally $\lab_\phi$ labels each vertex with its type (connective, modal operator, fixed point, literal or variable). Finally the relation $B_\phi$ contains a back edge from each vertex labelled with a variable to the successor of the fixed point binding that variable.

A \emph{partial function} $f : M \part N$ is a function $f' : M' \to N$ for some $M' \subseteq M$.
For a partial function $f : M \part N$ we denote by
\begin{multline*}
f[b_1/a_1, \dots, b_m/a_m, -/a_{m+1}, \dots, -/a_{m+n}] := \\ (f \setminus \{(a_i, b) \mid i \in \{1, \dots, m+n\}, b \in N\}) \cup \{(a_i, b_i) \mid i \in \{1, \dots, m\}\},
\end{multline*}	
the partial function, where values for $a_1, \dots a_m \in M$ are set to $b_1, \dots, b_m \in N$ respectively, and the values for $a_{m+1}, \dots a_{m+n} \in M$ are set as undefined.

We add some features to pointed Kripke-models for the game. A \emph{clocked model} is a tuple $(\A, w, c, a)$, where $(\A, w)$ is a pointed Kripke-model, $c: \Var \part \kappa$ and $a \in \{\new, \old\}$. Here $\kappa$ is a fixed cardinal larger than the size of the domain of $\A$. The partial function $c$ associates to each fixed point a clock to show how many times the model can return to that fixed point. As clocked models traverse a graph in the game, we use the identifier $\old$ to keep track of where they have been previously. We suppress the age identifier $a$ from the notation in cases where the distinction between new and old models does not matter. 


For simplicity we use the symbols $w$ and $c$ extensively and they should be read as ``the distinguished point and clocks of the model currently discussed'' throughout the rest of the paper.

Let $\AAA = (\A, w, \tup, a)$ be a clocked model and $\AA$ a set of clocked models. We redefine the following notations from the $\ML$ case for clocked models:
\begin{itemize}
	\item $\bo\AAA = \bo(\A, w, \tup, a) := \{(\A, w', \tup, a) \mid wR^\A w'\}$,
	\item $\bo\AA := \bigcup\limits_{\AAA \in \AA} \bo\AAA$.
	\item Let $f : \AA \to \bo\AA$ be a function such that $f(\AAA) \in \bo\AAA$ for every $\AAA \in \AA$. Then $\di_f\AA := f(\AA)$.
\end{itemize}
As for the $\ML$-game, the $\bo$-notation denotes the set of all successors of a single clocked model or a set of clocked models. The clocks and age identifier are inherited. The set $\di_f \AA$ contains one successor for each clocked model in $\AA$, given by the function $f$. 

Now let $\AA$ be a set of clocked models, $\AA_0$ a set of pointed models and $a \in \{\new, \old\}$. We use the following new notations:
\begin{itemize}
	\item $\poin(\AA) := \{(\M, w) \mid (\M, w, c, a) \in \AA \}$,
	\item $\clock(\AA_0) := \{ (\M, w, \emptyset, \new) \mid (\M, w) \in \AA_0\}$,
	\item $\AA_a := \{(\A, w, c, b) \in \AA \mid b = a\}$,
	\item $a(\AA) := \{(\A, w, c, a) \mid (\A, w, c, b) \in \AA \text{ for some } b \in \{\new, \old\}\}$.	
\end{itemize}
The set $\poin(\AA)$ contains the underlying pointed models of all clocked models in $\AA$ and the set $\clock(\AA_0)$ is the set of minimal clocked models with underlying pointed models from $\AA_0$. For an age identifier $a$, the set $\AA_a$ gives all clocked models in $\AA$ with that identifier and the set $a(\AA)$ gives all the models in $\AA$ with the age identifier changed to $a$. 

We define for $\L_\mu$ the standard approximant formulas that evaluate a fixed point only up to a bound. These approximants are formulas of \emph{infinitary} $\L_\mu$, where infinite conjunctions and disjunctions are allowed.

\begin{definition}
	Let $\alpha$ be an ordinal and $\psi(X)$ a formula of infinitary $\L_\mu$. Then the \emph{approximant formulas} $\mu^\alpha X. \psi(X)$ and $\nu^\alpha X. \psi(X)$ are defined by recursion as follows:
	\begin{itemize}
		\item $\mu^0 X. \psi(X) = \bot$ and $\nu^0 X. \psi(X) = \top$,
		\item $\mu^{\alpha+1}X.\psi(X) = \psi(\mu^\alpha X. \psi(X))$ and $\nu^{\alpha+1}X.\psi(X) = \psi(\nu^\alpha X. \psi(X))$,
		\item $\mu^\lambda X .\psi(X) = \bigvee\limits_{0 < \alpha < \lambda}\mu^\alpha X. \psi(X)$ and $\nu^\lambda X .\psi(X) = \bigwedge\limits_{0 < \alpha < \lambda}\nu^\alpha X. \psi(X)$ for a limit ordinal $\lambda$.
	\end{itemize}
\end{definition}

This definition differs from the usual one (see e.g. \cite{lenzi}) in that we leave out the $\bot$ disjunct and $\top$ conjunct in the limit ordinal cases, and more importantly, we do not necessarily approximate all fixed points so the resulting formula is not necessarily in infinitary $\ML$ but instead in infinitary $\L_\mu$. During the game we approximate several fixed points at once, starting from a specific point in the formula.
We define our own approximate formulas to reflect this.
\begin{definition}
	Let $\phi \in \L_\mu$. Let $T_\phi = (\set_\phi, \rel_\phi, \back_\phi, \lab_\phi)$ be the syntax tree with back edges of $\phi$ and let $s \in \set_\phi$. Let $s_1, \dots, s_n$ be the fixed point nodes above $s$ in $T_\phi$ in order with $s_1$ being the outermost and $s_n$ the innermost, and let $\lab_\phi(s_i) = \eta_i X_i$ for each $i \in \{1, \dots, n\}$. Let $c : \Var \part \kappa$ be a partial function with $\dom(c) = \{X_1, \dots, X_n\}$.
	
	The \emph{$(c, s)$-approximant} of $\phi$, $\phi^c_s$, is defined recursively as follows:
	\begin{itemize}
		\item if $\lab(s) = l \in \Lit$, then $\phi_s^c = l$,
		\item if $\lab(s) = \nabla \in \{\lor, \land\}$ and $s_1, s_2$ are the successors of $s$, then $\phi_s^c = \phi_{s_1}^c \nabla \phi_{s_2}^c$, 
		\item if $\lab(s) = \Delta \in \{\di, \bo\}$ and $s_1$ is the successor of $s$, then $\phi_s^c = \Delta \phi_{s_1}^c$, 
		\item if $\lab(s) = \eta X$, where $\eta \in \{\mu, \nu\}$ and $X \in \Var$, and $s_1$ is the successor of $s$, then $\phi_s^c = \eta X.\phi_{s_1}^{c}$,
		\item if $\lab(s) = X \in \Var \setminus \dom(c)$, then $\phi_s^c = X$.
		\item Let $\lab(s) = X_i$ and let $u$ be the $B_\phi$-successor of $s$. 
		
		If $c(X_i) = 0$, then if $\lab(s_i) = \mu X_i$, $\phi_s^c = \bot$ and if $\lab(s_i) = \nu X_i$, $\phi_s^c = \top$.
		
		If $c(X_i) = \alpha + 1$ for some ordinal $\alpha$, let $c_\alpha = c[\alpha/X_i, -/X_{i+1}, \dots, -/X_n]$. Now $\phi_s^c = \phi_u^{c_\alpha}$. 
		
		If $c(X_i)$ is a limit ordinal, let $c_\alpha = c[\alpha/X_i, -/X_{i+1}, \dots, -/X_n]$ for every $\alpha < c(X_i)$. Now 
		\[
		\phi_s^c = \bigvee\limits_{\alpha < c(X_i)} \phi_u^{c_\alpha} \text{ if } \lab(s_i) = \mu X_i \text{ \quad and \quad } \phi_s^c = \bigwedge\limits_{\alpha < c(X_i)} \phi_u^{c_\alpha} \text{ if } \lab(s_i) = \nu X_i.
		\]	
	\end{itemize}
\end{definition}

The formulas $\phi_s^c$ can contain infinite conjunctions and disjunctions but if all clocks are finite, then $\phi_s^c$ is an $\L_\mu$-formula. For instance if all models considered are finite, then finite clocks suffice. The following lemma formalizes the relationship of our approximant with the usual one.

\begin{lemma}\label{kellot2}
	Let $\phi \in \L_\mu$ and let $s$ be a vertex in the syntax tree of $\phi$ with $\lab(s) = \eta X$, where $\eta \in \{\mu, \nu\}$. Let $s_1$ be the successor of $s$ and let $c : \Var \part \kappa$ be a partial function with $X \notin \dom(c)$. Let $c_\alpha = c[\alpha/X]$. Now
	\[
	\phi_{s_1}^{c_\alpha} = \eta^{\alpha+1}X.\phi_{s_1}^c(X).
	\]
\end{lemma}
\begin{proof}
	Since $\eta^{\alpha+1}X.\phi_{s_1}^c(X) = \phi_{s_1}^c(\eta^\alpha X.\phi_{s_1}^c(X))$, where the parentheses notation refers to substituting free occurrences of $X$ with a formula, we may rewrite the claim in the form $\phi_{s_1}^{c_\alpha} = \phi_{s_1}^c(\eta^\alpha X.\phi_{s_1}^c(X))$. We show this by transfinite induction on $\alpha$.
	\begin{itemize}
		\item If $\alpha = 0$, then it is easy to see that $\phi_{s_1}^{c_\alpha} = \phi_{s_1}^c(\xi) = \phi_{s_1}^c(\eta^0 X.\phi_{s_1}^c)$, where $\xi = \bot$ if $\eta = \mu$ and $\xi = \top$ if $\eta = \nu$.
		
		\item Let $\alpha = \beta + 1$. By induction hypothesis $\phi_{s_1}^{c_\beta} = \phi_{s_1}^c(\eta^\beta X.\phi_{s_1}^c(X)) = \eta^\alpha X.\phi_{s_1}^c(X)$. We show by induction on the definition of $\phi_{t}^{c_\alpha}$, where $t$ is below $s$ in the syntax tree of $\phi$, that $\phi_{t}^{c_\alpha} = \phi_{t}^c(\eta^\alpha X.\phi_{s_1}^c(X))$. We first note that as $s_1$ is the successor of the fixed point node $s$, the fixed point of $X$ is the innermost one in $\dom(c_\alpha)$.
		\begin{itemize}
			\item If $\lab(t) = l \in \Lit$, then $\phi_{t}^{c_\alpha} = l = \phi_{t}^c(\eta^\alpha X.\phi_{s_1}^c(X))$.
			
			\item If $\lab(t) = Y \in \Var \setminus \dom(c_\alpha)$, then $\phi_{t}^{c_\alpha} = Y = \phi_{t}^c(\eta^\alpha X.\phi_{s_1}^c(X)).$
			
			\item Let $\lab(t) = Y \in \dom(c_\alpha) \setminus \{X\}$ and let $u$ be the $B$-successor of $t$. Now for some $c'$, $\phi_{t}^{c_\alpha} = \phi_u^{c'}$. Since the fixed point of $X$ is inside that of $Y$, $\phi_u^{c'}$ contains no free occurrences of $X$. Thus $\phi_u^{c'} = \phi_{t}^c(\eta^\alpha X.\phi_{s_1}^c(X))$.
			
			\item If $\lab(t) = X$, then $\phi_{t}^{c_\alpha} = \phi_{s_1}^{c_\beta}$. Because $\eta X$ is the innermost approximated fixed point, no clocks need to be reset. Since $\phi_t^c = X$ and by the induction hypothesis on $\alpha$, 
			\[
			\phi_t^{c_\alpha} = \phi_{s_1}^{c_\beta} = \eta^\alpha X.\phi_{s_1}^c(X) = \phi_{t}^c(\eta^\alpha X.\phi_{s_1}^c(X)).
			\]
			
			\item If $\lab(t) = \nabla \in \{\lor, \land\}$ and $t_1$ and $t_2$ are the successors of $t$, then by induction hypothesis, $\phi_{t_1}^{c_\alpha} = \phi_{t_1}^c(\eta^\alpha X.\phi_{s_1}^c(X))$ and $\phi_{t_2}^{c_\alpha} = \phi_{t_2}^c(\eta^\alpha X.\phi_{s_1}^c(X))$. Now 
			\begin{align*}
			\phi_{t}^{c_\alpha} &= \phi_{t_1}^{c_\alpha} \nabla \phi_{t_2}^{c_\alpha} = \phi_{t_1}^c(\eta^\alpha. X\phi_{s_1}^c(X)) \nabla \phi_{t_2}^c(\eta^\alpha X.\phi_{s_1}^c(X)) \\
			&= (\phi_{t_1}^c \nabla \phi_{t_2}^c)(\eta^\alpha X.\phi_{s_1}^c(X)) = \phi_t^c(\eta^\alpha X.\phi_{s_1}^c(X)).
			\end{align*}
			
			\item If $\lab(t) = \Delta \in \{\di, \bo\}$ and $t_1$ is the successor of $t$, then by induction hypothesis we have $\phi_{t_1}^{c_\alpha} = \phi_{t_1}^c(\eta^\alpha X.\phi_{s_1}^c(X))$. Thus 
			\[
			\phi_t^{c_\alpha} = \Delta \phi_{t_1}^{c_\alpha} = \Delta \phi_{t_1}^c(\eta^\alpha X .\phi_{s_1}^c(X)) = (\Delta \phi_{t_1}^c)(\eta^\alpha X.\phi_{s_1}^c(X)) = \phi_t^c(\eta^\alpha X.\phi_{s_1}^c(X)).
			\]
			
			\item If $\lab(t) = \eta_1 Y$, where $\eta_1 \in \{\mu, \nu\}$ and $Y \in \Var$, and $t_1$ is the successor of $t$, then by induction hypothesis $\phi_{t_1}^{c_\alpha} = \phi_{t_1}^c(\eta^\alpha X.\phi_{s_1}^c(X))$. Thus 
			\[
			\phi_t^{c_\alpha} = \eta_1 Y. \phi_{t_1}^{c_\alpha} = \eta_1 Y. \phi_{t_1}^c(\eta^\alpha X.\phi_{s_1}^c(X)) = (\eta_1 Y. \phi_{t_1}^c)(\eta^\alpha X.\phi_{s_1}^c(X)) = \phi_t^c(\eta^\alpha X.\phi_{s_1}^c(X)).
			\]
		\end{itemize}
		\item Now let $\alpha$ be a limit ordinal. By induction hypothesis $\phi_{s_1}^{c_\beta} = \phi_{s_1}^c(\eta^\beta X.\phi_{s_1}^c(X))$ for all $\beta < \alpha$. 
		The induction on $\phi_{t}^{c_\alpha}$ is handled in the same way as in the previous case with the exception of the $X$-case.
		\begin{itemize}
			\item If $\lab(t) = X$, then $\phi_{t}^{c_\alpha} = \mnabla\limits_{\beta < \alpha}\phi_{s_1}^{c_\beta}$, where $\mnabla = \bigvee$ if $\eta = \mu$ and $\mnabla = \bigwedge$ if $\eta = \nu$. By the induction hypothesis on $\alpha$, $\phi_{s_1}^{c_\beta} = \phi_{s_1}^c(\eta^\beta X.\phi_{s_1}^c(X))$ for all $\beta < \alpha$ so 
			\begin{align*}
			\phi_{t}^{c_\alpha} &= \mnabla\limits_{\beta < \alpha}\phi_{s_1}^c(\eta^\beta X.\phi_{s_1}^c(X)) 
			= \mnabla\limits_{\beta < \alpha}\eta^{\beta+1} X.\phi_{s_1}^c(X) 
			= \mnabla\limits_{0 < \beta < \alpha}\eta^\beta X. \phi_{s_1}^c(X) \\
			&= \eta^\alpha X. \phi_{s_1}^c(X)
			= \phi_t^c(\eta^\alpha X. \phi_{s_1}^c(X)).\qedhere
			\end{align*}
		\end{itemize}
	\end{itemize}
\end{proof}

\subsection*{The definition of the $\mathbf{\L_\mu}$-game}

Let $\Phi$ be a fixed finite set of propositional symbols. The formula size game for $\L_\mu(\Phi)$, $\FS_k(\AA_0, \BB_0)$, has two players, S (Samson) and D (Delilah). The game has as parameters two sets of pointed $\Phi$-models, $\AA_0$ and $\BB_0$, and a natural number $k$. S wants to show that the sets $\AA_0$ and $\BB_0$ can be separated with a $\L_\mu(\Phi)$-formula of size at most $k$. D on the other hand wants to show this is not possible. During the game S constructs step by step the syntax tree of a formula that, he claims, separates the sets. The number $k$ is a resource that is spent when S adds vertices to the syntax tree. If the resource $k$ ever runs out, S loses the game. S has to simultaneously show how the models in $\AA_0$ make the formula true and how the models in $\BB_0$ make it false. Each model traverses the incomplete syntax tree in a fashion similar to semantic games. The role of D is to keep S honest by deciding which branch of the tree she wants to see next. 

The modal $\mu$-calculus has the special feature of fixed point formulas. In terms of models traversing the syntax tree of a formula, the truth of a least fixed point $\mu X$ comes down to the model having to eventually stop returning to that fixed point. Thus, when entering such a fixed point, S must set a clock for each model, that shows how many more times he will return the model to that fixed point. On the other hand, to show a $\mu X$-formula is false in a model, the model would have to keep returning to the fixed point forever. Here it is the responsibility of D to declare how many returns is enough for her to be satisfied that the formula is indeed false. We now present the quite complex formalization of the game.

Let $\AA_0$ and $\BB_0$ be sets of pointed models and let $k_0 \in \Nset$. Let $\set^*$ be a predefined infinite set of vertices. The formula size game $\FS_{k_0}(\AA_0, \BB_0)$ for the modal $\mu$-calculus has two players, S and D. The positions of the game are of the form $P = (\set, \rel, \back, \lab, \res, \lef, \rig, v)$. Here $\set \subseteq \set^*$ and $(\set,\rel, \back)$ is a tree with back edges. The partial function 
\[
\lab : \set \part \{\land, \lor, \di, \bo\} \cup \Var \cup \{\mu X, \nu X \mid X \in \Var \} \cup \Lit(\Phi)
\]
assigns a label to some vertices of the tree.
The function $\res : \set \to \Nset$ assigns to each vertex the remaining resource, i.e. an upper bound for the size of the subformula starting from the vertex. The function $\lef: \set \to \powerset(\AA_0^*)$ assigns to each $v \in \set$ its left set of clocked models $\lef(v)$. Here $\AA_0^*$ contains all the clocked models obtainable from models of $\AA_0$ by altering the distinguished point, clocks and age identifier. Similarly, $\rig: \set \to \powerset(\BB_0^*)$ assigns the right set $\rig(v)$. The clock function of each model is of the form $c: \Var \part \kappa$, where $\kappa$ is a fixed cardinal larger than the size of the domain of any model in $\AA_0 \cup \BB_0$. Finally the vertex $v \in \set$ is the \emph{current vertex} of the position. We will always assume that the position $P$ has components with these names and $P'$ always consists of the same components with primes.


The starting position of the game is 
\[
(\{v_0\}, \emptyset, \emptyset, \emptyset, \{(v_0, k_0)\}, \{(v_0, \clock(\AA_0))\}, \{(v_0, \clock(\BB_0))\}, v_0).
\]

The first move is always D choosing finite subsets $\AA \subseteq \clock(\AA_0)$ and $\BB \subseteq \clock(\BB_0)$. The following position is
\[
(\{v_0\}, \emptyset, \emptyset, \emptyset, \{(v_0, k_0)\}, \{(v_0, \AA)\}, \{(v_0, \BB)\}, v_0).
\]

Throughout the whole game, D wins if at any position $P$, $\res(v) = 0$. D also wins if S is unable to make the choices required by a move.
Assume the game is in position $P$ and let $\lef(v) = \AA$, $\rig(v) = \BB$ and $\res(v) = k > 0$. We define two cases by whether $v$ already has a label or not. In each case we denote the following position by $P'$. 

\paragraph*{$v \notin \dom(\lab)$:}

S has a choice of eight different moves. Note that in this case $\AA = \AA_\new$ and $\BB = \BB_\new$.
\begin{itemize}
	\item{\bf $\lor$-move:} S chooses sets $\AA_1, \AA_2 \subseteq \AA$ s.t. $\AA_1 \cup \AA_2 = \AA$, and numbers $0< k_1, k_2 \leq k$ s.t. $k_1 + k_2 + 1 = k$. Then D chooses a number $i \in \{1,2\}$.
	Let $\set' = \set \cup \{v_1, v_2\}$, $\rel' = \rel \cup \{(v,v_1), (v,v_2)\}$, $B' = B$, 
	$\lab' = \lab[\lor/v]$,
	$\lef' = \lef[\AA_1/v_1, \AA_2/v_2, \old(\AA)/v]$, 
	$\rig' = \rig[\BB/v_1, \BB/v_2, \old(\BB)/v]$,
	$\res' = \res[k_1/v_1, k_2/v_2]$ and $v' = v_i$, where $v_1$ and $v_2$ are new vertices.
	
	\item{\bf $\land$-move:} Same as the $\lor$-move with the roles of $\AA$ and $\BB$ switched.
	
	\item{\bf $\di$-move:} S chooses a function $f : \AA \to \bo\AA$ such that $f(\AAA) \in \bo\AAA$ for each $\AAA \in \AA$. Then D chooses finite subsets $\AA' \subseteq \di_f\AA$ and $\BB' \subseteq \bo\BB$.
	Let $V' = V \cup \{v'\}$, $E' = E \cup \{(v, v')\}$, $B' = B$,
	$\lab' = \lab[\di/v]$,
	$\lef' = \lef[\AA'/v', \old(\AA)/v]$, $\rig' = \rig[\BB'/v', \old(\BB)/v]$ and $\res' = \res[k-1/v']$, where $v'$ is a new vertex.
	
	\item{\bf $\bo$-move:} Same as the $\di$-move with the roles of $\AA$ and $\BB$ switched.
	
	\item{\bf $\mu X$-move:} S chooses a variable $X \in \Var$ and for every $\AAA = (\A_\AAA, w_\AAA, \tup_\AAA) \in \AA$ an ordinal $\alpha_\AAA$. Then D chooses for every $\BBB = (\B_\BBB, w_\BBB, c_\BBB) \in \BB$ an ordinal $\alpha_\BBB$. Let $\tup'_\AAA = \tup_\AAA[\alpha_\AAA/X]$ and let $\AA' = \{(\A_\AAA, w_\AAA, \tup'_\AAA) \mid \AAA \in \AA\}$. Let $c'_\BBB = c_\BBB[\alpha_\BBB/X]$ and let $\BB' = \{(\B_\BBB, w_\BBB, c'_\BBB) \mid \BBB \in \BB\}$. Let $\set' = \set \cup \{v'\}$, $\rel' = \rel \cup \{(v, v')\}$, $B' = B$, $\lab' = \lab[\mu X/v]$, 
	$\lef' = \lef[\AA'/v', \old(\AA)/v]$, $\rig' = \rig[\BB'/v', \old(\BB)/v]$ and $\res' = \res[k-1/v']$.

	\item{\bf $\nu X$-move:} Same as the $\mu$-move with the roles of $\AA$ and $\BB$ switched.
	
	\item{\bf $X$-move:} S chooses $X \in \Var$. Let $u \in \set$ be the closest vertex above $v$ with $\lab(u) \in \{\mu X, \nu X\}$. If no such vertex exists, D wins the game. Otherwise if $\AA_\new = \BB_\new = \emptyset$, S wins the game. 
	
	Let $v'$ be the successor vertex of $u$. Let $\set' = \set$, $\rel' = \rel$, $\back' = \back \cup \{(v, v')\}$,
	$\lab' = \lab[X/v]$ and $\res' = \res$. 
	
	Assume that $\lab(u) = \mu X$. If $\tup(X) = 0$ for some $(\A, w, c) \in \AA$, D wins the game. 	
	Otherwise for every $\AAA = (\A_\AAA, w_\AAA, c_\AAA, \new) \in \AA$, S chooses $\alpha_\AAA < \tup_\AAA(X)$. 
	
	Let $\BB_+ = \{(\B, w, c, \new) \in \BB \mid c(X) \neq 0\}$. For every $\BBB = (\B_\BBB, w_\BBB, c_\BBB, \new) \in \BB_+$, D chooses $\alpha_\BBB < c_\BBB(X)$. 
	
	Let $Y_1, \dots, Y_n$ be the variables for which there is a node $t_i$ on the path from $u$ to $v$ with $\lab(t_i) \in \{\mu Y_i, \nu Y_i\}$. Let $\tup'_\AAA = \tup_\AAA[\alpha_\AAA/X, -/Y_1, \dots, -/Y_n]$ and let $\AA' = \{(\A_\AAA, w_\AAA, c'_\AAA, \new) \mid \AAA \in \AA\}$. Similarly let $c'_\BBB = c_\BBB[\alpha_\BBB/X, -/Y_1, \dots, -/Y_n]$ and $\BB' = \{(\B_\BBB, w_\BBB, c'_\BBB, \new) \mid \BBB \in \BB_+\}$. Let $\lef' = \lef[\AA'\cup\lef(v')/v', \old(\AA)/v]$ and $\rig' = \rig[\BB'\cup\lef(v')/v', \old(\BB)/v]$. 
	
	The case $\lab(u) = \nu X$ is the same with the roles of $\AA$ and $\BB$ switched.

	\item{\bf $\Lit$-move:} S chooses a $\Phi$-literal $\lit$. Let $\lab' = \lab[\lit/v]$ and let $\Delta' = \Delta$ for every other component $\Delta$. In the following position $P'$, if $\lit$ separates $\AA$ and $\BB$, then S wins the game. Otherwise, D wins.
\end{itemize}
\paragraph*{$v \in \dom(\lab)$:}
In this case S must perform the move dictated by $\lab(v)$ without creating any new vertices. These moves are essentially performed only on new models. We again denote the following position by $P'$ and in each case we have $\nabla' = \nabla$ for $\nabla \in \{\set, \rel, \back, \lab, \res\}$.
\begin{itemize}
	\item If $\lab(v) = \lor$, then let $v_1$ and $v_2$ be the successors of $v$. S chooses sets $\AA_1, \AA_2 \subseteq \AA_\new$ s.t. $\AA_1 \cup \AA_2 = \AA_\new$. Then D chooses a number $i \in \{1,2\}$. Let 
	$\lef' = \lef[\AA_1 \cup \lef(v_1)/v_1, \AA_2 \cup \lef(v_2)/v_2, \old(\AA)/v]$, $\rig' = \rig[\BB \cup \rig(v_1)/v_1, \BB \cup \rig(v_2)/v_2, \old(\BB)/v]$ and $v' = v_i$. 
	
	\item The case $\lab(v) = \land$ is the same as $\lor$ with the roles of $\AA$ and $\BB$ switched.
	
	\item If $\lab(v) = \di$, then let $v'$ be the successor of $v$. S chooses a function $f : \AA_\new \to \bo\AA_\new$ such that $f(\AAA) \in \bo\AAA$ for each $\AAA \in \AA_\new$. Then D chooses finite $\AA' \subseteq \di_f\AA_\new$ and $\BB' \subseteq \bo\BB_\new$. Let 
	$\lef' = \lef[\AA'\cup\lef(v')/v_1, \old(\AA)/v]$ and $\rig' = \rig[\BB'\cup\rig(v')/v_1, \old(\BB)/v]$. 
	
	\item The case $\lab(v) = \bo$ is the same as $\di$ with the roles of $\AA$ and $\BB$ switched.
	
	\item If $\lab(v) = \mu X$ for some $X \in \Var$, then let $v'$ be the successor of $v$. for every $\AAA = (\A_\AAA, w_\AAA, \tup_\AAA) \in \AA_\new$ an ordinal $\alpha_\AAA$. Then D chooses for every $\BBB = (\B_\BBB, w_\BBB, c_\BBB) \in \BB_\new$ an ordinal $\alpha_\BBB$. Let $\tup'_\AAA = \tup_\AAA[\alpha_\AAA/X]$ and let $\AA' = \{(\A_\AAA, w_\AAA, \tup'_\AAA) \mid \AAA \in \AA_\new\}$. Let $\tup'_\BBB = \tup_\BBB[\alpha_\BBB/X]$ and let $\BB' = \{(\B_\BBB, w_\BBB, \tup'_\BBB) \mid \BBB \in \BB_\new\}$. Let $\lef' = \lef[\AA'\cup\lef(v')/v', \old(\AA)/v]$ and $\rig' = \rig[\BB'\cup\rig(v')/v', \old(\BB)/v]$. 
	
	\item If $\lab(v) = X \in \Var$, then let $v'$ be the $B$-successor of $v$. The rest is very similar to the unlabelled $X$-move; the only differences are that the move is again essentially only performed on new models and the condition for an immediate win for S is $\AA_\new = \BB_\new = \emptyset$.
	
\end{itemize}

Note that just like in the $\ML$-game, the $\di$-move cannot be performed if $\bo \AAA = \emptyset$ for some $\AAA \in \AA$, and dually for the $\bo$-move.

Figures \ref{kuva1} and \ref{kuva2} show how the partial syntax tree $(\set, \rel, \back, \lab)$ changes when a $\land$-move or an $X$-move is made. The current vertex $v$ is highlighted. The numbers near the vertices show the remaining resource given by $\res$. In addition each vertex has a left and a right set of clocked models given by $\lef$ and $\rig$. Of these only the new models affected by the move are shown in the pictures.

\begin{figure}[ht]
	\centering
	\begin{tikzpicture}
	\node[circle, inner sep = 1pt, minimum size = 21pt, draw, very thick, label={left:8}] (1) at (0,3) {$\mu X$};
	\node[circle, inner sep = 1pt, minimum size = 20pt, draw, very thick, label={left:7}] (2) at (0,1.75) {$\lor$};
	\node[circle, inner sep = 1pt, minimum size = 20pt, draw, very thick, label={above:4}] (3) at (-1.5,.75) {$\land$};
	\node[circle, inner sep = 1pt, minimum size = 20pt, draw, very thick, label={left:2}] (4) at (1.5,.75) {$\di$};
	\node[circle, inner sep = 1pt, minimum size = 20pt, draw, very thick, label={left:1}] (5) at (1.5,-.5) {$X$};
	\node[circle, inner sep = 1pt, minimum size = 20pt, draw, very thick, label={above:1}, label={left:$\AA$}, label={right:$\BB\Smash{_1}$}] (6) at (-2.5,-.5) {};
	\node[circle, inner sep = 1pt, minimum size = 20pt, draw, very thick, label={above:2}, label={left:$\AA$}, label={right:$\BB\Smash{_2}$}] (7) at (-.5,-.5) {};
	\node[circle, inner sep = 1pt, minimum size = 16pt, draw, very thick] (7') at (-.5,-.5) {};
	
	\draw[->, thick] (1) to (2);
	\draw[->, thick] (2) to (3);
	\draw[->, thick] (2) to (4);
	\draw[->, thick] (4) to (5);
	\draw[->, thick] (3) to (6);
	\draw[->, thick] (3) to (7);
	
	\draw[->, thick] (5) .. controls (3, 1) and (1.5, 3) .. (2);
	
	\node[circle, inner sep = 1pt, minimum size = 21pt, draw, very thick, label={left:8}] (1) at (-8,3) {$\mu X$};
	\node[circle, inner sep = 1pt, minimum size = 20pt, draw, very thick, label={left:7}] (2) at (-8,1.75) {$\lor$};
	\node[circle, inner sep = 1pt, minimum size = 20pt, draw, very thick, label={above:4}, label={left:$\AA$}, label={right:$\BB$}] (3) at (-1.5-8,.75) {};
	\node[circle, inner sep = 1pt, minimum size = 20pt, draw, very thick, label={left:2}] (4) at (1.5-8,.75) {$\di$};
	\node[circle, inner sep = 1pt, minimum size = 20pt, draw, very thick, label={left:1}] (5) at (1.5-8,-.5) {$X$};
	\node[circle, inner sep = 1pt, minimum size = 16pt, draw, very thick] (3') at (-1.5-8,.75) {};
	
	\draw[->, thick] (1) to (2);
	\draw[->, thick] (2) to (3);
	\draw[->, thick] (2) to (4);
	\draw[->, thick] (4) to (5);

	\draw[->, thick] (5) .. controls (3-8, .5) and (2-8, 2) .. (2);
	
	\node (nuoli) at (-4, 1) {\huge $\Rightarrow$};
	
	\end{tikzpicture}
	\caption{}
	\label{kuva1}
\end{figure}

\begin{figure}[ht]
	\centering
	\begin{tikzpicture}
	\node[circle, inner sep = 1pt, minimum size = 21pt, draw, very thick, label={left:8}] (1) at (0,3) {$\mu X$};
	\node[circle, inner sep = 1pt, minimum size = 20pt, draw, very thick, label={below:7}, label={left:$\AA'$}, label={right:$\BB'$}] (2) at (0,1.75) {$\lor$};
	\node[circle, inner sep = 1pt, minimum size = 20pt, draw, very thick, label={above:4}] (3) at (-1.5,.75) {};
	\node[circle, inner sep = 1pt, minimum size = 20pt, draw, very thick, label={left:2}] (4) at (1.5,.75) {$\di$};
	\node[circle, inner sep = 1pt, minimum size = 20pt, draw, very thick, label={below:1}] (5) at (1.5,-.5) {$X$};
	
	\draw[->, thick] (1) to (2);
	\draw[->, thick] (2) to (3);
	\draw[->, thick] (2) to (4);
	\draw[->, thick] (4) to (5);

	\draw[->, thick] (5) .. controls (3, 1) and (1.5, 3) .. (2);
	
	\node[circle, inner sep = 1pt, minimum size = 21pt, draw, very thick, label={left:8}] (1) at (-8,3) {$\mu X$};
	\node[circle, inner sep = 1pt, minimum size = 20pt, draw, very thick, label={below:7}] (2) at (-8,1.75) {$\lor$};
	\node[circle, inner sep = 1pt, minimum size = 20pt, draw, very thick, label={above:4}, label={left:$\AA$}] (3) at (-1.5-8,.75) {};
	\node[circle, inner sep = 1pt, minimum size = 20pt, draw, very thick, label={left:2}] (4) at (1.5-8,.75) {$\di$};
	\node[circle, inner sep = 1pt, minimum size = 20pt, draw, very thick, label={below:1}, label={left:$\AA$}, label={right:$\BB$}] (5) at (1.5-8,-.5) {};
	\node[circle, inner sep = 1pt, minimum size = 16pt, draw, very thick] (5') at (1.5-8,-.5) {};
	
	\draw[->, thick] (1) to (2);
	\draw[->, thick] (2) to (3);
	\draw[->, thick] (2) to (4);
	\draw[->, thick] (4) to (5);

	\node (nuoli) at (-4, 1) {\huge $\Rightarrow$};
	
	\end{tikzpicture}
	\caption{}
	\label{kuva2}
\end{figure}
An important feature of our game is that, even though it contains infinite branching, every single play of the game is still finite.

\begin{lemma}\label{finite}
	Every play of the game $\FS_k(\AA, \BB)$ is finite.
\end{lemma}
\begin{proof}
	A play could be infinite only if at least one variable is reached infinitely many times. Of these variables let $X$ be the one with the outmost fixed point. Every time $X$ is reached, if $\lef(v)_\new = \rig(v)_\new = \emptyset$, S wins and otherwise, the clock of at least one model is lowered. There are only finitely many models at any given position, since D always chooses finite subsets of models after modal moves. Since clocks are inherited by successor models in modal moves and ordinals are well-founded, eventually either a clock of S will reach $0$ and D will win or $X$ will be reached with empty sets of models and S wins. 
\end{proof}

For the essential theorem about how the game works, we assume that the strategy of S is uniform. This essentially means that S has a formula in mind and he follows the structure of that formula when constructing the syntax tree during the game.

\begin{definition}
	Let $\phi \in \L_\mu$ and let $T_\phi = (\set_\phi, \rel_\phi, \back_\phi, \lab_\phi)$ be the syntax tree with back edges of $\phi$. Let $P = (\set, \rel, \back, \lab, \res, \lef, \rig, v)$ be a position in a game $\FS_k(\AA, \BB)$. 
	
	A function $g : \set \to \set_\phi$ is a \emph{position embedding} if it satisfies the following conditions:
	\begin{enumerate}
		\item $g(v_0)$ is the root of $T_\phi$, where $v_0$ is the vertex of the starting position,
		\item $g$ is an embedding of $(\set, \rel)$ to $(\set_\phi, \rel_\phi)$,
		\item $g \upharpoonright \dom(\lab)$ is an embedding of $(\set, \back, \lab)$ to $(\set_\phi, \back_\phi, \lab_\phi)$,
		\item for each $u \in \set$, $\size(\phi_{g(u)}) \leq \res(u)$.
	\end{enumerate}
	
	Let $\delta$ be a strategy of S. We say that \emph{$\delta$ follows $\phi$ from position $P$ (via the function $g$)} if there is a position embedding $g : \set \to \set_\phi$ such that for each position $P'$ reachable from $P$ via the strategy $\delta$, the function $g$ can be extended to a position embedding $g' : \set' \to \set_\phi$.
	
	Finally $\delta$ is \emph{uniform} if $\delta$ follows a formula $\phi \in \L_\mu$ from the starting position.
\end{definition}

We are now ready to prove that the game indeed works as we intended. In the following if there is a vertex with label $\mu X$ we shall call $X$ a \emph{$\mu$-variable}, and if there is one labelled $\nu X$, we call $X$ a \emph{$\nu$-variable}. Note that we assume all fixed points have separate variables.

\begin{theorem}\label{peruslause2}
	Let $\AA_0$ and $\BB_0$ be sets of pointed models and let $k \in \Nset$. Then the following conditions are equivalent:
	\begin{enumerate}
		\item S has a uniform winning strategy in the game $\FS_k(\AA_0, \BB_0)$.
		\item There is a sentence $\phi \in \L_\mu(\Phi)$ s.t. $\phi$ separates $\AA_0$ and $\BB_0$ and $\size(\phi) \leq k$.
	\end{enumerate}
\end{theorem}
\begin{proof}
	$(2) \Rightarrow (1)$. Let $\phi \in \L_\mu(\Phi)$ be a sentence such that $\phi$ separates $\AA_0$ and $\BB_0$ and $s(\phi) \leq k$. Let $T_\phi = (\set_\phi, \rel_\phi, \back_\phi, \lab_\phi)$ be the syntax tree with back edges of $\phi$.  The strategy of S is to follow the structure of $T_\phi$ when forming $(V, E, B, \lab)$, to use the resource $k$ accordingly and to choose maximal appropriate sets of models when necessary.

	If $P = (\set, \rel, \back, \lab, \res, \lef, \rig, v)$ is a position, let $\AA = \lef(v)$, $\BB = \rig(v)$ and $k = \res(v)$. We define the strategy more precisely and simultaneously prove by induction that the strategy is uniform and the following condition holds for every position $P$ of the game:
	\begin{equation} \tag{$\ast$} \label{ehto2}
	\begin{aligned} 
	&(\A, w) \vDash \phi_{g(v)}^c \text{ for every } (\A, w, c) \in \AA_\new \text{ and } \\
	&(\B, w) \nvDash \phi_{g(v)}^c \text{ for every } (\B, w, c) \in \BB_\new,
	\end{aligned}
	\end{equation}
	where $g$ is a position embedding showing the uniformity of the strategy. 
	
	In the starting position, we set $g(v_0)$ as the root of $T_\phi$. We note that $\size(\phi_{g(v_0)}) = \size(\phi) \leq k$ by assumption. Since there are no clocks yet, $\phi_{g(v_0)}^c = \phi$ for every clocked model and since the sentence $\phi$ separates the sets $\AA_0$ and $\BB_0$, (\ref{ehto2}) holds no matter which subsets $\AA \subseteq \clock(\AA_0)$ and $\BB \subseteq \clock(\BB_0)$ D chooses.
	
	We now divide the proof into cases based on whether $v$ already has a label or not. We choose the move for $S$ according to the label of $g(v)$. We only treat one of each pair of dual cases.
	
	\paragraph*{$v \notin \dom(\lab):$}
	\begin{itemize} 
		\item $\lab_\phi(g(v)) = l \in \Lit(\Phi)$: Then by induction hypothesis $l$ separates the sets $\AA$ and $\BB$ so S wins by making the corresponding $\Lit$-move.
		
		\item $\lab_\phi(g(v)) = \lor$: By induction hypothesis, (\ref{ehto2}) holds for this position so for every $(A, w, c) \in \AA$, $(A, w) \vDash \phi_{g(v)}^c$. Now $\phi_{g(v)}^c = \phi_{s_1}^c \lor \phi_{s_2}^c$, where $s_1, s_2 \in \set_\phi$ are the successors of $g(v)$, so $(\A, w) \vDash \phi_{s_1}^c \lor \phi_{s_2}^c$. Let $\AA_1 = \{(\A, w, \tup) \in \AA \mid (\A, w) \vDash \phi_{s_1}^c\}$ and $\AA_2 = \{(\A, w, \tup) \in \AA \mid (\A, w) \vDash \phi_{s_2}^c\}$. On the other side, for every $(B, w, c) \in \BB$, $(B, w) \nvDash \phi_{s_1}^c \lor \phi_{s_2}^c$ so $(B, w) \nvDash \phi_{s_1}^c$ and $(B, w) \nvDash \phi_{s_2}^c$. We set $g' = g[v_1/s_1, v_2/s_2]$, where $v_1, v_2$ are the new vertices in $\set$. Now (\ref{ehto2}) holds in both of the possible following positions. Let $k_1 = \size(\phi_{s_1})$ and $k_2 = k - k_1 - 1$. Since $\size(\phi_{g(v)}) \leq k$, $\size(\phi_{s_2}) = \size(\phi_{g(v)}) - \size(\phi_{s_1}) - 1 \leq k - k_1 - 1 = k_2$.
		
		\item $\lab_\phi(g(v)) = \di$: By induction hypothesis, for every $(\A, w, c) \in \AA$, $(\A, w) \vDash \phi_{g(v)}^c$. Since $\phi_{g(v)}^c = \di \phi_{s_1}^c$, where $s_1$ is the successor of $g(v)$, $(\A, w) \vDash \di \phi_{s_1}^c$. Thus there is $(\A, w', c) \in \bo\AA$ s.t. $(\A, w') \vDash \phi_{s_1}^c$. Let $f : \AA \to \bo\AA$ be a function mapping every $(\A, w, \tup)$ to such a $(\A, w', \tup)$. Now $(\A, w') \vDash \phi_{s_1}^c$ for every $(\A, w', c) \in \di_f\AA$. On the other side, for every $(\B, w, c) \in \BB$, since $(\B, w) \nvDash \di \phi_{s_1}^c$, for every $(\B, w', c) \in \bo(\B, w, c)$ we get $(\B, w') \nvDash \phi_{s_1}^c$. Thus $(\B, w') \nvDash \phi_{s_1}^c$ for every $(\B, w', c) \in \bo\BB$ so (\ref{ehto2}) holds in the next position no matter which subsets $\AA' \subseteq \di_f\AA$ and $\BB' \subseteq \bo\BB$ D chooses. For uniformity we set $g' = g[v'/s_1]$, where $v'$ is the new vertice. Now $\size(\phi_{s_1}) = \size(\phi_{g(v)}) - 1 \leq k - 1$.
		
		\item Let $\lab_\phi(g(v)) = \mu X$: By induction hypothesis, for every $\AAA = (\A, w, c) \in \AA$, $(\A, w) \vDash \phi_{g(v)}^c$. Since $\phi_{g(v)}^c = \mu X.\phi_{s_1}^c(X)$, where $s_1$ is the successor of $g(v)$, $(\A, w) \vDash \mu X.\phi_{s_1}^c(X)$. Thus there is an ordinal $\alpha$ such that $(\A, w) \vDash \mu^{\alpha+1} X. \phi_{s_1}^c(X)$. S chooses $\alpha_\AAA = \alpha$ as the new clock. By Lemma~\ref{kellot2}, $(\A, w) \vDash \phi_{s_1}^{c_\alpha}$. On the other side, by induction hypothesis, for every $\BBB = (\B, w, c) \in \BB$, $(\B, w) \nvDash \mu X.\phi_{s_1}^c(X)$ so for every ordinal $\alpha$, $(\B, w) \nvDash \mu^{\alpha+1} X.\phi_{s_1}^c(X)$. Thus no matter which ordinal $\beta$ D chooses, we get $(\B, w) \nvDash \mu^{\beta+1} X.\phi_{s_1}^c(X)$ and by Lemma~\ref{kellot2}, $(\B, w) \nvDash \phi_{s_1}^{c_\beta}$. Therefore (\ref{ehto2}) holds in the following position. Uniformity is proved in the same way as in the $\di$-case.	
		
		\item $\lab_\phi(g(v)) = X$, where $X \in \Var$: Let $u$ be the $B$-successor of $g(v)$. 
		
		Assume that $X$ is a $\mu$-variable. Now by induction hypothesis, for every $\AAA = (\A, w, c) \in \AA$,  $(\A, w) \vDash \phi_{g(v)}^c$. There are three cases according to $c(X)$. 
		\begin{enumerate}
			\item If $c(X) = 0$, we get a contradiction since then $\phi_{g(v)}^c = \bot$.
			
			\item If $c(X) = \alpha + 1$ for some $\alpha$, then S chooses $\alpha$ as the new clock for $X$ in $(\A, w, c)$. Now $\phi_{g(v)}^c = \phi_u^{c_\alpha}$ so $(\A, w) \vDash \phi_u^{c_\alpha}$. 
			
			\item If $c(X)$ is a limit ordinal, then 
			\[
			\phi_{g(v)}^c = \bigvee\limits_{\alpha < c(X)} \phi_u^{c_\alpha} \text{ \quad so \quad } (\A, w) \vDash \bigvee\limits_{\alpha < c(X)} \phi_u^{c_\alpha}.
			\]
			S chooses the new clock $\alpha$ such that $(\A, w) \vDash \phi_u^{c_\alpha}$ holds. 
		\end{enumerate}
		
		On the other side, by induction hypothesis, for every $\BBB = (\B, w, c) \in \BB$, $(\B, w) \nvDash \phi_{g(v)}^c$. We again have three cases.
		\begin{enumerate}
			\item If $c(X) = 0$, $\BBB$ will be removed from the game and can be disregarded. 
			
			\item Let $c(X) = \alpha + 1$. Now $\phi_{g(v)}^c = \phi_u^{c_\alpha}$ so $(\B, w) \nvDash \phi_u^{c_\alpha}$. By Lemma~\ref{kellot2}, $(\B, w) \nvDash \mu^{\alpha+1} X.\phi_u^{c'}(X)$, where $c' = c_\alpha[-/X]$. Let $\beta \leq \alpha$ be the choice of D for the new clock. Now by monotonicity, $(\B, w) \nvDash \mu^{\beta+1} X.\phi_u^{c'}(X)$. We use Lemma~\ref{kellot2} again and obtain $(\B, w) \nvDash \phi_u^{c_\beta}$. 
			
			\item Finally let $c(X)$ be a limit ordinal. Now
			\[
			\phi_{g(v)}^c = \bigvee\limits_{\alpha < c(X)} \phi_u^{c_\alpha} \text{ \quad so \quad } (\B, w) \nvDash \bigvee\limits_{\alpha < c(X)} \phi_u^{c_\alpha}.
			\]
			Thus $(\B, w) \nvDash \phi_u^{c_\alpha}$ for any $\alpha < c(X)$ D chooses.
		\end{enumerate}
		Uniformity is trivial here since no new vertices were created.
		
		The case where $X$ is a $\nu$-variable is the same with the roles of $\AA$ and $\BB$ switched.
	\end{itemize}
	
	\paragraph*{$v \in \dom(\lab):$}
	
	The moves are essentially the same as in the unlabelled case. The main differences are that the type of the move is already determined by $\lab(v)$, and the resource splittings are already fixed. In disjunction and conjunction moves new models can be left to wait in the branch not chosen by D as the following position. We will consider only this special case of waiting new models here.
	
	$\lab(v) = \lor$: S chooses the sets $\AA_1$ and $\AA_2$ of new models as in the unlabelled case. There may however be some new models present in $v_1$ or $v_2$. If so, these models are there because of previous $\lor$-moves, for the first of which $v$ had no label. By induction hypothesis and the unlabelled case, (\ref{ehto2}) held for both of the possible following positions and therefore $(\A, w) \vDash \phi_{s_i}^c$ for every model $(\A, w, c)$ in the corresponding left model set. Inductively we see that $(\A, w) \vDash \phi_{s_i}^c$ for every $(\A, w, c) \in \lef(v_i)$. The same argument shows that $(\B, w) \nvDash \phi_{s_i}^c$ for every $(\B, w, c) \in \rig(v_i)$. Thus (\ref{ehto2}) holds for the sets $\AA_i \cup \lef(v_i)$ and $\BB \cup \rig(v_i)$ in both of the possible following positions of position $P$.
	
	\medskip
	
	$(1) \Rightarrow (2)$. Let $\delta$ be a uniform winning strategy for S. Let $\phi \in \L_\mu(\Phi)$ be the formula $\delta$ follows. We denote the position embedding showing the uniformity of $\delta$ in each position by $g$. By Lemma~\ref{finite} every play of the game is finite so the game tree induced by the strategy $\delta$ is well-founded. We prove by well-founded induction on the game positions reachable with $\delta$ that the same condition (\ref{ehto3}) as above holds in every position of the game. 
	\begin{equation} \tag{$\ast$} \label{ehto3}
	\begin{aligned} 
	&(\A, w) \vDash \phi_{g(v)}^c \text{ for every } (\A, w, c) \in \AA_\new \text{ and } \\
	&(\B, w) \nvDash \phi_{g(v)}^c \text{ for every } (\B, w, c) \in \BB_\new,
	\end{aligned}
	\end{equation}
	
	In a position $P$ reachable with $\delta$, let $\lef(v) = \AA$, $\rig(v) = \BB$ and $\res(v) = k$. We again consider the unlabelled and labelled case separately and only treat one of each pair of dual moves.
	
	\paragraph*{$v \notin \dom(\lab):$}
	\begin{itemize}	
		\item $\lab(g(v)) \in \Lit(\Phi)$: Since $\delta$ follows $\phi$, the next move according to $\delta$ is a $\Lit$ move choosing that literal. Since $\delta$ is a winning strategy, that literal separates the sets $\AA$ and $\BB$ and (\ref{ehto3}) holds.
		
		\item $\lab(g(v)) = \lor$: Let $s_1$ and $s_2$ be the successors of $g(v)$. Let $\AA_1$, $\AA_2$, $k_1$ and $k_2$ be the selections of S according to $\delta$. By induction hypothesis (\ref{ehto3}) holds in both possible following positions so for every $(\A, w, c) \in \AA_i$, $(\A, w) \vDash \phi_{s_i}^c$ for $i \in \{1, 2\}$. Since $\AA = \AA_1 \cup \AA_2$, for every $(\A, w, c) \in \AA$, $(\A, w) \vDash \phi_{s_1}^{c} \lor \phi_{s_2}^{c}$. In addition $\phi_{s_1}^{c} \lor \phi_{s_2}^{c} = \phi_{g(v)}^c$ so $(\A, w) \vDash \phi_{g(v)}^c$. Let $(\B, w, c) \in \BB$. By (\ref{ehto3}) in the following positions, $(\B, w) \nvDash \phi_{s_1}^c$ and $(\B, w) \nvDash \phi_{s_2}^c$ so $(\B, w) \nvDash \phi_{s_1}^c \lor \phi_{s_2}^c$. Since  $\phi_{s_1}^c \lor \phi_{s_2}^c = \phi_{g(v)}^c$, $(\B, w) \nvDash \phi_{g(v)}^c$. Thus (\ref{ehto3}) holds in $P$.
		
		\item $\lab(g(v)) = \di$: Let $s_1$ be the successor of $g(v)$. Let $f : \AA \to \bo\AA$ be the function chosen by S according to $\delta$. Let $(\A, w, c) \in \AA$. By induction hypothesis (\ref{ehto3}) holds in the following position no matter which subsets of $\di_f\AA$ and $\bo\BB$ D chooses so for $f(\A, w, c) = (\A, w', c) \in \di_f \AA$, $(\A, w') \vDash \phi_{s_1}^c$. Since $w'$ is a successor of $w$, now $(\A, w) \vDash \di \phi_{s_1}^c$. In addition $\di \phi_{s_1}^c = \phi_{g(v)}^c$ so $(\A, w) \vDash \phi_{g(v)}^c$. Let $(\B, w, c) \in \BB$. By induction hypothesis (\ref{ehto3}) holds for all possible following positions so for every $(\B, w', c) \in \bo\BB$, $(\B, w') \nvDash \phi_{s_1}^c$. Therefore $(\B, w) \nvDash \di \phi_{s_1}^c$ and so $(\B, w) \nvDash \phi_{g(v)}^c$. Thus (\ref{ehto3}) holds in $P$.

		\item $\lab(g(v)) = \mu X$: Let $s_1$ be the successor of $g(v)$. Let $\AAA = (\A, w, c) \in \AA$ and let $\alpha_\AAA = \alpha$ be the choice of S according to $\delta$. By induction hypothesis, (\ref{ehto3}) holds in the following position so $(\A, w) \vDash \phi_{s_1}^{c_\alpha}$. By Lemma~\ref{kellot2}, $(\A, w) \vDash \mu^{\alpha+1} X.\phi_{s_1}^c$. Thus $(\A, w) \vDash \mu X. \phi_{s_1}^c$. Since $\mu X. \phi_{s_1}^c = \phi_{g(v)}^c$, $(\A, w) \vDash \phi_{g(v)}^c$. On the other side, by induction hypothesis, for every $\BBB = (\B, w, c) \in \BB$ and for any choice $\alpha$ of D for the new clock, $(\B, w) \nvDash \phi_{s_1}^{c_\alpha}$. Thus by Lemma~\ref{kellot2}, $(\B, w) \nvDash \mu^{\alpha+1}X.\phi_{s_1}^c(X)$ for every $\alpha < \kappa$. Since $\kappa > \card(\B)$, this means that $(\B, w) \nvDash \mu X.\phi_{s_1}^c$ and so $(\B, w) \nvDash \phi_{g(v)}^c$. Thus (\ref{ehto3}) holds in $P$.
		
		\item $\lab(g(v)) = X$: If $\AA_\new = \BB_\new = \emptyset$, S wins and (\ref{ehto3}) trivially holds. Let $u$ be the $B$-successor of $g(v)$.
		Assume that $X$ is a $\mu$-variable and let $(\A, w, c) \in \AA$. 
		We have three cases according to the ordinal~$c(X)$.
		\begin{enumerate}	
			\item If $c(X) = 0$, D wins the game, which is a contradiction, since $\delta$ is a winning strategy for S. 
			\item Assume that $c(X) = \alpha +1$. Let $\beta \leq \alpha$ be the choice of S for the new clock according to $\delta$. By induction hypothesis (\ref{ehto3}) holds in the following position so $(\A, w) \vDash \phi_u^{c_\beta}$. By Lemma~\ref{kellot2}, $(\A, w) \vDash \mu^{\beta+1} X.\phi_u^{c'}$, where $c' = c_\beta[-/X]$. Thus by monotonicity, $(\A, w) \vDash \mu^{\alpha+1} X.\phi_u^{c'}$. By Lemma~\ref{kellot2} again, $(\A, w) \vDash \phi_u^{c_\alpha}$ and so $(\A, w) \vDash \phi_{g(v)}^c$. 
			\item Now assume $c(X)$ is a limit ordinal and let $\alpha < c(X)$ be the choice of S according to $\delta$. Now by induction hypothesis
			$(\A, w) \vDash \phi_u^{c_\alpha}$. Thus 
			\[
			(\A, w) \vDash \bigvee\limits_{\alpha < c(X)} \phi_u^{c_\alpha}
			\]
			so $(\A, w) \vDash \phi_{g(v)}^c$.
		\end{enumerate}
		For every $(\B, w, c) \in \BB$, regardless of the choice of D for the new clock, (\ref{ehto3}) holds in the following position. We again have three cases.
		\begin{enumerate}
			\item If $c(X) = 0$, then since $(\B, w) \nvDash \bot$, we get $(\B, w) \nvDash \phi_{g(v)}^c$.
			\item If $c(X) = \alpha + 1$, then $\alpha$ is a choice available to D so $(\B, w) \nvDash \phi_u^{c_\alpha}$. Thus $(\B, w) \nvDash \phi_{g(v)}^c$. 
			\item If $c(X)$ is a limit ordinal, every $\alpha < c(X)$ is a choice available to D so $(\B, w) \nvDash \phi_u^{c_\alpha}$ for every $\alpha < c(X)$. Thus 
			\[
			(\B, w) \nvDash \bigvee\limits_{\alpha < c(X)}\phi_u^{c_\alpha}.
			\]
			Therefore $(\B, w) \nvDash \phi_{g(v)}^c$ so (\ref{ehto3}) holds in $P$.
		\end{enumerate}
	\end{itemize}
	
	\paragraph*{$v \in \dom(\lab):$}
	
	All the moves in this case are proved the same way as in the unlabelled case. Note that since (\ref{ehto3}) refers only to new models, it trivially holds for terminal positions where $\AA_\new = \BB_\new = \emptyset$ for an $X$-move.
	
	For the very first move of the game, where D chooses finite subsets of the original sets of clocked models $\clock(\AA_0)$ and $\clock(\BB_0)$, by induction hypothesis (\ref{ehto3}) holds in the following position no matter which subsets D chooses. Therefore all models in $\AA_0$ and $\BB_0$ also satisfy the condition (\ref{ehto3}). Since there are no clocks in the starting position, this means that $\phi$ separates the sets $\AA_0$ and $\BB_0$. By the uniformity of $\delta$, $\size(\phi) = \size(\phi_{g(v_0)}) \leq \res(v_0) = k$.
\end{proof}
Note that condition (\ref{ehto2}) does not depend on old models and so we do not refer to them in this proof. We add old models to the game to make the proof of Lemma~\ref{bisimilar} in the next section easier.

Unlike other similar theorems, Theorem~\ref{peruslause2} has the added requirement of uniformity for the strategy of S. We conjecture that the theorem would still hold even without this condition, but proving this has turned out to be difficult. Note, however, that to prove undefinability results for $\L_\mu(\Phi)$, one need only define a winning strategy for D, and so the uniformity of strategies for S need not be considered. Note further that if a property of $\Phi$-models is not definable in $\L_\mu(\Phi)$, then clearly it is not definable in $\L_\mu$.

\section{Succinctness of FO over $\mathbf{\L_\mu}$}

We move on to the definitions and lemmas needed to show that $\FO$ is non-elementarily more succinct than $\L_\mu$. We need a lemma similar to Lemma~\ref{bisim} that gives D a winning strategy if bisimilar models are produced on both sides of a vertex. In the case of the $\L_\mu$-game, the clocks of the clocked models must also be taken into account. We define a sufficient condition for clocked models to be useful for D in the game and call them \nice\ models.

\begin{definition}
	The \emph{depth} of a pointed finite tree model $(\M, w)$, $d(\M, w)$, is the length of a maximal path of transitions in the model starting from $w$.
\end{definition}

\begin{definition}
	In a position $P$ of a game $\FS_k(\AA, \BB)$, let $u$ be a vertex and let $\MMM = (\M, w, c, a)$ be a clocked finite tree model in $\lef(u) \cup \rig(u)$. We say the model $\MMM$ is \emph{\nice}, if for every $X \in \dom(c)$, the clock of D is equal to or greater than the depth of the model, i.e.
	\begin{itemize}
		\item if $\MMM \in \lef(u)$, for every $\nu$-variable $X \in \dom(c)$, $c(X) \geq d(\M, w)$,
		\item if $\MMM \in \rig(u)$, for every $\mu$-variable $X \in \dom(c)$, $c(X) \geq d(\M, w)$.
	\end{itemize}
	We also say the model $\MMM$ is \emph{strictly \nice} if the above condition holds for strict inequality $>$ instead of $\geq$.
\end{definition}

We prove the analogue of Lemma~\ref{bisim} for \nice\ clocked models.

\begin{lemma}\label{bisimilar}
	Let $P$ be a position of a game $\FS_k(\AA, \BB)$. If there are strictly \nice\ clocked models $(\A, w_\A, c_\A, \new) \in \lef(v)$ and $(\B, w_\B, c_\B, \new) \in \rig(v)$ such that $(\A, w_\A)$ and $(\B, w_\B)$ are bisimilar finite tree models, then D has a winning strategy from position $P$.
\end{lemma}
\begin{proof}
	We show that D can maintain a slightly modified condition where we only require the models to be \nice\ and allow one of the two models to be old. For $\lor$- and $\land$-moves D need only choose the side for which the two models are both present. For modal moves we see by bisimilarity that no matter which successor S chooses, a bisimilar model will be present on the opposite side. Moreover, the depth of the models is decreased and the clocks are inherited so the condition is maintained. For new fixed points D need only set her clock to be equal to the depth of the models. On a literal move, D will win since the bisimilar models cannot be separated by a literal.
	
	For $X$-moves, let $u$ be the vertex the game returns to. If this is the first time since position $P$ the game returns to $u$, D will lower the clock and since the models are strictly \nice, D can now decrease the clock to the same value as the depth. If on the other hand there has been a previous return to $u$, then there are two cases. If there has been a modal move in between this and the previous return, then the depth has decreased and D will decrease the clock to the same value. If there have been no modal moves, the pointed models have not changed and since we allow one of the models to be old, D will now consider the old version of her model, with a clock larger by one, instead of the new one. Consider the position $P'$ right after D switches a new model for an old one in this fashion. Assume by symmetry that this model is $\BBB$ on the right side. Consider the path from the vertex just returned to, $u$, to the vertex $s$ where the $X$-move was made. If there are no vertices with label $\lor$ on this path, then D can just follow this path to $s$ until the clock $c_\A(X)$ runs out and D wins. Assume there are some vertices with label $\lor$ on the path. For each of those vertices, the child that is not on the path from $u$ to $s$ has a new version of $\BBB$ in the right set. This is because the model $\BBB$ has passed through the disjunction before the $X$-move and models on the right side are always copied on both sides of a disjunction. If S splits the left model $\AAA$ away from the path to $s$, then D will consider the new copy of $\BBB$ from then on. If $\AAA$ stays on the path from $u$ to $s$ indefinitely, D wins when the clock $c_\A(X)$ runs out.
	
	If D uses this strategy, eventually S will either make a literal move and lose, or a clock of S for one of the bisimilar models will eventually run out. In either case, D wins.
\end{proof}

We want to use the same graph based invariant for the proof as we did for the $\ML$-case. The only question that remains is: which models should determine the graph of the current vertex $v$? In other words, which models is S claiming to be able to separate in each position of the game? Certainly the models in $\lef(v)$ and $\rig(v)$ should be included, but they will not suffice, since other models can already be below $v$ and thus involved with the subformula beginning from $v$. We define a way for D to collect all the models in the tree below a vertex $s$ to see which models are, in a sense, "currently in $s$". We define the collected sets separately for the left and the right side. Recall that $\poin(\AA)$ is the set of underlying pointed models of clocked models in $\AA$.

\begin{definition}
Let $P$ be a position such that $\lab(v)$ is not a literal and no vertices have label $\di$ or $\bo$. For each vertex $s \in \set$, we define the \emph{left collection of $s$ in $P$}, denoted by $\LL(s)$, recursively starting from the leaves of $(\set, \rel, \back)$:
\begin{itemize}
	\item if $s$ is an unlabelled leaf, then $\LL(s) = \poin(\lef(s))$,
	\item if $\lab(s)$ is a $\mu$-variable, then $\LL(s) = \emptyset$,
	\item if $\lab(s)$ is a $\nu$-variable, then $\LL(s) = \poin(\lef(s))$,
	\item if $\lab(s) = \eta X$ and the successor is $s_1$, then $\LL(s) = \poin(\lef(s)_{\new}) \cup \LL(s_1)$,
	\item if $\lab(s) = \lor$ and the successors are $s_1, s_2$, then $\LL(s) = \poin(\lef(s)_{\new}) \cup \LL(s_1) \cup \LL(s_2)$
	\item if $\lab(s) = \land$ and the successors are $s_1, s_2$, then $\LL(s) = \poin(\lef(s)_{\new}) \cup (\LL(s_1) \cap \LL(s_2))$.
\end{itemize}
\end{definition}

The \emph{right collection of $s$ in $P$}, $\RR(s)$, is defined symmetrically with $\lef$ and $\rig$, as well as $\mu$ and $\nu$, switched at every point of the definition.

Note that since $v$ does not have a literal label, a leaf of $(\set, \rel, \back)$ can only be either an unlabelled vertex or a vertex with a variable label. This is because any vertex labelled with an operator will have at least one successor and the game ends in the next position after any literal move. 

In the following, we will associate superscripted sets like $\LL'(s)$ with the position $P'$ with the same superscript just like we have done so far with the components of the position.

\begin{lemma}\label{lrsets}
	Let $P^\circ$ be a position in a game $\FS_k(\AA, \BB)$, where the current vertex $u = v^\circ$ is a successor of a fixed point vertex and no modal moves are made. Let $P'$ be a position after $P^\circ$ such that no $X$-move has returned to a vertex above $u$ since $P^\circ$. Then $\LL^\circ(u) \subseteq \LL'(u)$ and $\RR^\circ(u) \subseteq \RR'(u)$.
\end{lemma}
\begin{proof}
	The proof proceeds by induction. We assume that for position $P$, $\LL^\circ(u) \subseteq \LL(u)$ and $\RR^\circ(u) \subseteq \RR(u)$ and we show that the inclusion also holds for the next position $P'$.
	
	If S makes a $\eta X$-move, the new models in $\lef(v)$ are moved to $\lef'(v')$ but they still remain in $\LL'(v)$ so $\LL'(u) = \LL(u)$. Note that clocks of the models do change but $\LL'(v)$ only looks at the underlying pointed models. For the same reasons $\RR'(u) = \RR(u)$.
	
	For $\lor$-moves the models in $\lef(v)$ are split among the successors $v_1$ and $v_2$ so they are still in $\LL'(v)$ in position $P'$. The models in $\rig(v)$ are copied to both $v_1$ and $v_2$ so they are still in $\RR'(v)$. Thus $\LL'(u) = \LL(u)$ and $\RR'(u) = \RR(u)$. The case of $\land$-moves is symmetric with the two sides switched everywhere.
	
	If S makes an $X$-move, it can either be a return to $u$ or to a vertex $s$ below $u$. Assume that the return is made to $u$ and that $X$ is a $\mu$-variable. Now the models in $\lef(v)$ are moved to $\lef'(u)$ so any that were already in $\LL(u)$ stay there and more may be added so $\LL(u) \subseteq \LL'(u)$. The models in $\rig(v)$ cease to be relevant in $\rig'(u)$ but they remain as old models in $\rig'(v)$ and are still counted for $\RR'(v)$ in $P'$ as they were in $P$ so $\RR'(u) = \RR(u)$. The case of a $\nu$ is symmetric with the two sides switched.
	
	Finally assume that an $X$-move is made returning to a vertex $s$ below $u$. Assume again that $X$ is a $\mu$-variable. The models in $\lef(v)$ are moved to $\lef'(s)$ so $\LL(s) \subseteq \LL'(s)$. Everything not below $s$ remains unchanged so $\LL(u) \subseteq \LL'(u)$. The models in $\rig(v)$ cease to be relevant in $\rig'(s)$ but they remain as old models in $\rig'(v)$ and are still counted for $\RR'(s)$ and therefore also for $\RR'(u)$ just like in $P$. Thus $\RR'(u) = \RR(u)$. The case of a $\nu$ is again symmetric.
\end{proof}

We finally have all of the required notation and lemmas to show that the non-elementary succinctness gap is present also between $\FO$ and modal $\mu$-calculus. 

\begin{theorem}\label{foandmu}
	First-order logic is non-elementarily more succinct than modal $\mu$-calculus.
\end{theorem}
\begin{proof}
	We prove an analogous result to Lemma~\ref{winstrat} for the $\L_\mu$ game. We use the notation $\G(\VV, \EE)$ for the same graph as in subsection \ref{graphs}. The precise statement we prove is as follows:
	
	Let $n \in \Nset$ and let $k_0 \in \Nset$. If $k_0 < \log ( \chi (\G(\CC_n, \DD_n)))$, then D has a winning strategy in the game $\FS_{k_0}(\CC_n, \DD_n)$.
	
	In this proof we only consider \nice\ models. Many positions in the game $\FS_{k_0}(\CC_n, \DD_n)$ have also non-\nice\ models but they are not needed for the strategy of D we describe here and can safely be ignored. We will assume all models are \nice\ and occasionally comment on why models remain or cease to be \nice\ after some moves of the game.
	
	We show by induction that D has a strategy to maintain the following condition in any position $P = (\set, \rel, \back, \lab, \res, \lef, \rig, v)$:
	\begin{equation*}\tag{$\ast$}\label{ehto}
		\res(v) < \log(\chi(\G(\LL(v), \RR(v)))
	\end{equation*}
	
	At the start of the game, \ref{ehto} holds by assumption. Since the sets $\CC_n$ and $\DD_n$ are already finite, D can keep the full sets for the first move of the game.
	
	We first show that if S ever makes a modal move while \ref{ehto} holds, D gets a winning strategy for the game. We assume $v \notin \dom(\lab)$ since the first modal move in a game must always be made in an unlabelled vertex. In this case there are no other vertices below $v$ so $\LL(v) = \lef(v)$ and $\RR(v) = \rig(v)$. We assume $\res(v) > 1$ so that S can make a modal move and not lose immediately due to the resource running out. From \ref{ehto} we obtain $\chi(\G(\lef(v), \rig(v))) > 2$ so there are \nice\ clocked models $(\liitos(\M_1, w_1), c_1, \new), (\liitos(\M_2, w_2), c_2, \new) \in \lef(v)$ and $(\liitos\{(\M_1, w_1), (\M_2, w_2)\}, c_3, \new) \in \rig(v)$. Now if S makes a $\di$- or $\bo$-move, then in the following position $P'$ there is $i \in \{1,2\}$ s.t. $(\M_i, w_i, c_i, \new) \in \lef(v')$ and $(\M_i, w_i, c_3, \new) \in \rig(v')$. As these two share the same underlying pointed model they are bisimilar and moreover, since the depth has decreased by at least $1$ from the previous position, the models are strictly \nice. By Lemma~\ref{bisimilar}, D now has a winning strategy from position $P'$.
	
	If S makes a $\lor$-move, let $v_1$ and $v_2$ be the successors of $v$. In the following position, whichever it may be, we have $\LL(v) = \LL(v_1) \cup \LL(v_2)$ and $\RR(v) = \RR(v_1) \cap \RR(v_2)$. Let $\G_s = \G(\LL(s), \RR(s)) = (V_s, E_s)$ for $s \in \{v, v_1, v_2\}$. We obtain $V_v = V_{v_1} \cup V_{v_2}$ and $E_v \upharpoonright V_{v_i} \subseteq E_{v_i}$ for $i \in \{1,2\}$. By Lemma~\ref{väritys}, 
	\[
	\chi(\G_v) \leq \chi(V_{v_1}, E_v \upharpoonright V_{v_1}) + \chi(V_{v_2}, E_v \upharpoonright V_{v_2}) \leq \chi(\G_{v_1}) + \chi(\G_{v_2}).
	\]
	Thus (just like in the proof of Theorem~\ref{winstrat}) we obtain $\res(v_i) < \log(\chi(\G_{v_i})$ for some $i \in \{1,2\}$ so \ref{ehto} holds in the following position after D chooses this $i$.
	
	If S makes a $\land$-move, let $v_1$ and $v_2$ be the successors of $v$. In the following position we have $\LL(v) = \LL(v_1) \cap \LL(v_2)$ and $\RR(v) = \RR(v_1) \cup \RR(v_2)$. We use the notation $\G_s = (V_s, E_s)$ from the previous case and obtain $V_v = V_{v_1} \cap V_{v_2}$ and $E_v = (E_{v_1}\upharpoonright V_v) \cup (E_{v_2}\upharpoonright V_v)$. By Lemma~\ref{väritys},
	\[
	\chi(G_v) \leq \chi(V_v, E_{v_1}\upharpoonright V_v)\chi(V_v, E_{v_2}\upharpoonright V_v) \leq \chi(\G_{v_1})\chi(\G_{v_2}).
	\]
	Thus we again obtain $\res(v_i) < \log(\chi(\G_{v_i})$ for some $i \in \{1,2\}$ so \ref{ehto} holds in the following position after D chooses this $i$.
	
	If S makes a $\eta X$-move, where $\eta \in \{\mu, \nu\}$, then D sets her clock for each model at the same value as the depth of the model. All \nice\ models remain \nice\ and \ref{ehto} is maintained.
	
	If S makes an $X$-move, by \ref{ehto}, S does not immediately win the game. Assume that $u$ is the vertex returned to and $P^\circ$ is the previous position when $u$ was the current vertex. Let $P'$ be the position after this $X$-move. By Lemma~\ref{lrsets} we obtain $\LL^\circ(u) \subseteq \LL'(u)$ and $\RR^\circ(u) \subseteq \RR'(u)$. By induction hypothesis, \ref{ehto} held in $P^\circ$ and clearly $\res'(u) = \res^\circ(u)$, so \ref{ehto} still holds in $P'$.
	
	If S makes a $\Lit$-move, by \ref{ehto}, $\lef(v) \neq \emptyset$ and $\rig(v) \neq \emptyset$. Since all the models are propositionally equivalent, D wins the game.
	
	By Theorem~\ref{peruslause2}, we obtain that there is no sentence $\phi \in \L_\mu(\emptyset)$ that separates $\CC_n$ from $\DD_n$ with $\size(\phi) \leq \log(\chi(\G(\CC_n, \DD_n))) = \tower(n-1)$. Thus $\FO$ is non-elementarily more succinct than $\L_\mu$.

\end{proof}

\begin{remark}
	Just like in the case of $\ML$, we remark that the result of Theorem~\ref{foandmu} also holds for DAG-size. This is again because the difference between the size of an $\L_\mu$ formula in our sense and the DAG-size of the same formula is at most exponential.
\end{remark}

\section{Conclusion}

We have defined formula size games for basic modal logic and modal $\mu$-calculus. The games utilize resource parameters to achieve a truly two-player game. In the case of modal logic the players only construct one branch of the game tree. This is in contrast with the original Adler-Immerman game, where the players form the whole tree in a single play. For modal $\mu$-calculus the recursive nature of fixed point operators necessitates returning to previously visited nodes. However, the game still traverses only one possible path through the game tree in a single play and some branches can remain unvisited for the entire play. The $\mu$-calculus game has infinite branching but the use of decreasing ordinal clocks, as in \cite{boundedGTS}, makes each play of the game finite.

We used the games to show that the property ``all successor models are $n$-bisimilar with each other'' cannot be defined in basic modal logic or modal $\mu$-calculus with a formula of size less than the exponential tower of height $n-1$. On the other hand this property can be defined in $\FO$ with a formula of size linear in $n$. This means that $\FO$ is non-elementarily more succinct than both $\ML$ and $\L_\mu$. We also show that the same property can be defined in two-dimensional modal logic $\ML^2$ with a formula of size exponential in $n$. Thus the non-elementary succinctness gap is also present between $\ML^2$ and both $\ML$ and $\L_\mu$ 

We find the $\ML$-game to be a useful tool for proving lower bounds on the size of $\ML$-formulas. Depending on the desired result, the game can also be modified to count a more specific parameter instead of formula size, such as the number or nesting depth of a specific operator. 

The $\mu$-calculus game is also functional for proving lower bounds, but with some caveats. The main theorem stating the usefulness of the game, Theorem~\ref{peruslause2}, requires uniform strategies for S. This means that we assume S has a single formula in mind and always plays according to that formula. It may be that this restriction could be removed but we have been unable to prove this. However, to show succinctness results we only need one direction of the equivalence so the issue is usually not relevant in practice. The greater concern is whether the game can be successfully used to prove succinctness results for $\mu$-calculus in more complicated contexts. Here we only generalize a result already obtained with the $\ML$ game and we have so far failed to produce any other results with the $\L_\mu$ game due to its sheer complexity. A question related to this difficulty would be whether the game could be simplified significantly while still preserving its functionality. It would be especially interesting to apply the game to open problems related to $\mu$-calculus and succinctness, such as whether there is a polynomial transformation from $\L_\mu$ to the guarded fragment or from vectorial form to regular~$\L_\mu$ \cite{10224192520150401}.


\bibliographystyle{abbrv}
\bibliography{formulasizejournal}

\end{document}